\documentclass[letterpaper,showpacs,twocolumn,prd,aps,10pt,amsmath,amssymb]{revtex4-1}
\pdfoutput=1
\usepackage{amsmath,amsthm}
\usepackage{bm}
\usepackage{graphicx}
\usepackage{tensor}
\usepackage[hyperindex,breaklinks]{hyperref}

\allowdisplaybreaks[1] 

\newtheorem{definition}{Definition}
\newtheorem{theorem}{Theorem}

  \newcommand{\llangle}{\langle\!\langle}
  \newcommand{\rrangle}{\rangle\!\rangle}
\renewcommand{\d}{\mathrm{d}}
  \newcommand{\oo}{\infty}

\renewcommand{\O}{\mathcal{O}}

  \newcommand{\del}{\partial}

\newsavebox\sarrbox%
\savebox{\sarrbox}{\includegraphics[scale=.75]{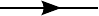}}%
\def\sarr{\usebox\sarrbox}%
\newsavebox\darrbox%
\savebox{\darrbox}{\includegraphics[scale=.75]{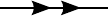}}%
\def\darr{\usebox\darrbox}%

\begin{document}

\title{Quantum astrometric observables I: time delay in classical and quantum gravity}

\author{Igor Khavkine}
\email{i.khavkine@uu.nl}
\affiliation{Institute for Theoretical Physics, Utrecht University\\
Leuvenlaan 4, NL--3584 CE Utrecht, The Netherlands}
\date{\today}

\begin{abstract}
A class of diffeomorphism invariant, physical observables, so-called
astrometric observables, is introduced. A particularly simple example, the time
delay, which expresses the difference between two initially synchronized
proper time clocks in relative inertial motion, is analyzed in detail.
It is found to satisfy some interesting inequalities related to the
causal structure of classical Lorentzian spacetimes. Thus it can serve
as a probe of causal structure and, in particular, of violations of
causality.  A quantum model of this observable as well as the
calculation of its variance due to vacuum fluctuations in quantum
linearized gravity are sketched. The question of whether the causal
inequalities are still satisfied by quantized gravity, which is
pertinent to the nature of causality in quantum gravity, is raised, but
it is shown that perturbative calculations cannot provide a definite
answer. Some potential applications of astrometric observables in
quantum gravity are discussed.
\end{abstract}

\maketitle

\section{Introduction}\label{intro}
The issue of physical observables in both classical and quantum theories
of gravity has been a topic of long standing interest both practically
and theoretically. Practically, precise models for tracking the positions
of objects on scales from Solar System to cosmological require input
from general relativity. Such models go under the generic name of
relativistic astrometry~\cite{soffel-astro}. Also, models of the very
early Universe rely on incorporating quantum gravitational effects in
order to predict potentially observable signatures in current
cosmological observations (Sec.~5.7 of~\cite{woodard-qg}). Theoretically,
the algebra of physical observables, with each observable a mathematical
model of an experimental outcome, is an integral part of a complete
classical or quantum theory of gravity~\cite{bergmann-obsv, rovelli-obsv, lp-obsv, dittrich-obsv, dittrich-tambo, pss-obsv}. In each of
these cases, one has to confront the problem that, while physical
observables are expected to be invariant under spacetime
diffeomorphisms, in the usual formulation of general relativity
everything is described in terms of tensor fields, which are covariant
but not \emph{invariant} under spacetime diffeomorphisms.

A resolution of the problem of physical observables would then involve
two parts. First, one has to explicitly describe a sufficiently large
class of spacetime diffeomorphism-invariant (or \emph{diff-invariant})
quantities expressed in terms of tensor fields. Second, one has to
identify elements of this class that correspond to outcomes of some
experiments of interest. The literature on this subject is
extensive~\cite{bergmann-obsv, rovelli-obsv, lp-obsv, dittrich-obsv, dittrich-tambo,
gmh-obsv, pss-obsv} (see also references therein), but no solution has
been entirely successful. To illustrate the difficulties, consider the
following two examples. A simple-to-describe class of diff-invariant
quantities consists of the spacetime integrals of the form
\begin{equation}
	\int_M f(g,\phi) \, v_g ,
\end{equation}
where $M$ is the entire spacetime, $f$ is some smooth spacetime scalar
defined only in terms of the metric $g$ and other dynamical fields $\phi$, and
$v_g$ is the metric volume form. Unfortunately, even ignoring the
issue of convergence of such integrals, this class of
diff-invariant quantities is not rich enough to describe the outcomes of
any experiments that we are likely to perform (since such experiments
would necessarily be localized in a finite region of spacetime). The
other example is more abstract. Consider (formally) the physical phase
space of general relativity defined as the quotient by spacetime
diffeomorphisms of the space of solutions of Einstein's equations.
Ostensibly, any function on this space is a diff-invariant quantity, and
hence a physical observable. Moreover, all physical observables are so
captured. However, due to the abstract nature of this construction, it
is not possible to assign a clear physical meaning to any element of
this class. There is, a priori, no effective way to specify an
individual element of this class or to carry out practical calculations
with it.

The aim of this work is to take a pragmatic approach to the explicit
construction of physical observables and apply it to more
theoretical problems like studying the causal structure of quantum
gravity. From a theoretical point of view, the abstract notion of a
physical observable, sketched in the previous paragraph, as a function
on the physical phase space is quite satisfactory. The main problem
remaining is to identify observables of interest and given them a
physical interpretation in terms of a modeled experimental outcome. A
natural way of addressing this difficulty, inspired by the methods used
in practical problems like relativistic astrometry, is to start with a
potential experiment in mind and construct a sufficiently detailed
mathematical model of it. Such a model should include sufficiently many
dynamical variables representing parts of the experimental apparatus
such that the desired measurement outcome can be modeled using the
relative configurations of these variables. The result is a mathematical
model of a measurement outcome, in other words a physical observable.
This observable, by virtue of its operational definition, should then be
a diff-invariant quantity and thus an element of the algebra of
functions on the physical phase space of the theory. Now though,
fortunately, since we started out by modeling an experiment, its
physical interpretation is clear.

A mathematical model of an experiment interacting with dynamical gravity
is likely to make reference to solutions of geodesic or wave equations
on unspecified (indeed dynamical) metric backgrounds. Coupled with the
large variety of experiments that could be imagined and modeled, a
remaining practical difficulty is that the resulting physical observable
is still specified only implicitly and may not be immediately amenable
to practical calculations. It appears that this problem can only be
overcome on a case by case basis. For the particular observable
considered in this work, the \emph{time delay}, this difficulty is overcome by
appealing to perturbation theory and providing explicit formulas, based
on an exact implicit definition, in terms of one-dimensional integrals
over linear metric perturbations about Minkowski space.

The idea of using operationally or ``relationally'' defined observables
in gravitational theories is not entirely novel. It has been previously
considered in~\cite{rovelli-rel, rovelli-qg, dittrich-obsv}.
Unfortunately, that work has remained at a rather abstract level and did
not make use of sufficiently realistic experimental models, thus keeping
the physical interpretation of the constructed observables somewhat
moot. The previous works that used ideas most similar to ours
are~\cite{ford-lightcone}, \cite{ohlmeyer}, and~\footnote{A.\ Roura and
D.\ Arteaga (private communication).}. Unfortunately, the original work
of~\cite{ford-lightcone} and its follow-ups~\cite{ford-top, ford-focus,
ford-angle}, while exhibiting a clear physical interpretation, left many
mathematical loose ends. In particular, the issues of diff-invariance
(or gauge invariance) and regularization were not treated entirely
satisfactorily, both of which are explicitly addressed in this work, see
Secs.~\ref{gauge-inv} and~\ref{dyn-appr}. Another work in a similar
spirit is~\cite{tsamis-woodard}, especially at the technical level,
though with a different physical motivation. On the other hand, the
original work~\cite{woodard-thesis} gives several different motivations
for the technical calculations, including an approach very similar to
that of this paper in terms of the construction of diff-invariant,
physically meaningful gravitational observables. At the technical level,
the main departure of this work from that of~\cite{woodard-thesis,
tsamis-woodard} is in the use of smeared observables to regularize
divergences appearing due to the use of geodesics of idealized,
point-like particles, see Sec.~\ref{fin-resol}.

In Sec.~\ref{op-def}, we operationally define the \emph{time delay} physical
observable (or rather a family of related observables). Section~\ref{class-model}
gives an exact, though implicit, mathematical model for this physical
observable in a theory of gravity coupled to a minimal amount of matter
modeling the experimental apparatus. Section~\ref{caus-ineq} contains an analysis of
why the time delay is an observable interesting for studying the causal
structure of gravity. In particular, two important inequalities are
derived directly from the Lorentzian character of the metric field.
Section~\ref{class-calc}, using technical results on the perturbative solution of
the geodesic and parallel transport equations presented in the \hyperref[pertsol]{Appendix},
gives an explicit formula for the time delay in linearized gravity.
Sections~\ref{quant-model} and~\ref{quant-calc} sketch how the time delay should be defined as a
quantum observable and how explicit calculations in linearized quantum
gravity can capture some aspects of causal structure of quantum gravity.
Due to the added complexity of quantum mechanics, these two sections are
naturally less detailed than the preceding ones. The issues discussed in
these sections will be addressed in more detail elsewhere.
Section~\ref{discussion} concludes with a discussion of the results and an outlook to
future work.

\section{Operational description and gauge invariance}\label{op-def}
The time delay observable is defined by the following experimental
protocol, Fig.~\ref{triang-geom-clock} (which is of course only an
idealization of a real experiment).  Consider a laboratory in inertial
motion (free fall). The laboratory carries a clock that measures the
proper time along its trajectory. The laboratory also carries an
orthogonal frame, which is parallel-transported along the lab's
worldline. (The frame could be Fermi-Walker-transported if the motion
were not inertial.) At a moment of the experimenter's choosing, the lab
ejects a probe in a predetermined direction, fixed with respect to the
lab's orthogonal frame and with a predetermined relative velocity. The
probe then continues to move inertially and carries its own proper time
clock. The two clocks are synchronized to $0$ at the ejection event $O$.
After ejection, the probe continuously broadcasts its own proper time
(time stamped signals), in all directions using an electromagnetic
signal (which hence travels at the speed of light). At a predetermined
proper time interval $s$ after ejection, event $Q$, the lab records the
probe signal and its emission time stamp $\tau(s)$, sent from event $P$.
Call $s$ the \emph{reception time}, $\tau(s)$ the \emph{emission time}
and the difference
\begin{equation}
	\delta\tau(s)=s-\tau(s)
\end{equation}
the \emph{time delay}.

\begin{figure}
\includegraphics{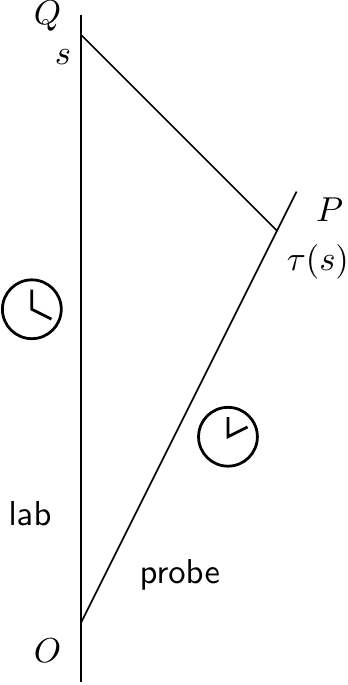}
\caption{%
Geometry of the experimental protocol defining the reception
time $s$, emission
time $\tau(s)$, and the time delay $\delta\tau(s)=s-\tau(s)$. The
synchronization/ejection point is $O$. The signal emission point is $P$
and the signal reception point is $Q$.
}
\label{triang-geom-clock}
\end{figure}

The time delay $\delta\tau(s)$ as well as the emission and reception
time are presumed to have been measured with negligible inaccuracy. Of
course, that is a severe idealization. For it to be reasonable, the
magnitude of $\delta\tau(s)$ must exceed the noise from the intrinsic
inaccuracies in the probe and lab instruments (clocks, gyroscopes,
ejection mechanism, transmission and recording uncertainties, etc.).
Given the smallness of both the classical and quantum contributions to
$\delta\tau(s)$ (which are suppressed by all of the following: magnitude
of light speed, smallness of spacetime curvature, and smallness of
$\hbar$), is unlikely to be reasonable in our own Universe, at least for
na\"ive ways to realize this experimental setup.

However, there is no a priori reason for not being able to perform such
measurements successfully with (significantly) more clever or improved
experimental techniques, or in a universe with different values of some
of the fundamental constants. A successful theory of quantum gravity
should be able to yield quantitative predictions (for any universe) for
this and related observables. Some set of these observables may actually
be practically measurable in our own Universe. As such, the time delay,
by virtue of its simplicity and ease of physical interpretation, serves
as a useful benchmark for dealing with whatever practical difficulties
are likely come up in calculations involving similar, but perhaps more
realistic, observables.

\section{Classical mathematical model}\label{class-model}
A significantly idealized mathematical model of the experiment described
in the preceding section, in the classical theory, consists of the
geometrical objects collected in the following definition.

\begin{definition}
	A \emph{lab-equipped} spacetime $(M,g,O,e^a_i)$ consists of an
	oriented, Lorentzian, time-oriented, globally hyperbolic,
	$n$-dimensional spacetime $(M,g)$, a point $O\in M$ and an oriented
	orthonormal frame $e^a_i\in T_OM$, with $a$ an abstract tensor index
	and $i=0,1,\ldots,n$, where $e^a_0$ is timelike and future-directed.
\end{definition}

See Sec.~2.4 of~\cite{wald} for the distinction between abstract and coordinate
tensor indices.
Physically, the point $O$ represents the spacetime event when the probe
is ejected from the lab. The vector $e^a_0$ is tangent to the lab's
worldline and $e^a_i$, $i=1,2,\ldots,n$ is the oriented spatial frame
carried by the lab. Note that the restrictions on the Lorentzian
geometry of $(M,g)$ may not all be necessary. Also, in this work we only
consider the case $n=4$.

Of course, due to the background independence of gravitational physics,
the measurements carried out in spacetimes related by diffeomorphisms
(i.e.,\ gauge transformations) must be identical. It is thus useful to
introduce the following notion of equivalence.

\begin{definition}
	Two lab-equipped spacetimes $(M,g,O,e^a_i)$ and $(M',g',O',e^{\prime
	a}_i)$ are \emph{gauge equivalent} if there exists a diffeomorphism
	$\chi\colon M\to M'$ such that $\chi_*g = g'$, $\chi(O) = O'$ and
	$\chi_* e^a_i = e^{\prime a}_i$, where $\chi_*$ denotes the
	differential push-forward.
\end{definition}

An observable is modeled mathematically by a function on the space of
lab-equipped spacetimes. The time delay observable is defined by
implementing the protocol outlined in the preceding section. First, we
need the following further definitions.
\begin{definition}
	The \emph{lab worldline}, $Q(s)=\exp_O(su)$, is the geodesic passing
	through $O$ with tangent vector $u^a=e^a_0$. Here $\exp_O\colon
	T_OM\to M$ is the usual geodesic exponential map.

	The \emph{probe worldline}, $P(t)=\exp_O(tv)$, is the geodesic
	passing through $O$ with tangent vector $v^a = v^i e^a_i$, with
	$v^i\in\mathbb{R}^{1,3}$ a timelike, future-directed, unit vector,
	chosen independent of the spacetime geometry.

	The \emph{signal worldline}, $Z(t,\lambda)$, is the null geodesic
	emanating from a point on the probe worldline, $Z(t,1)=P(t)$, and
	intersecting the lab worldline, $Z(t,0)=Q(s)$, with the earliest
	possible $s$ (alternatively, if $s$ is fixed, then $t$ is chosen to be
	the latest possible).
\end{definition}
The values of $t$ and $s$ connected by $Z$ are functionally related. This
relationship defines the time delay observable.
\begin{definition}
	If $t$ and $s$ are such that $P(t)$ and $Q(s)$ are connected by
	$Z(t,\lambda)$, they are referred to as a pair of \emph{emission} and
	\emph{reception} times. The functional relationship between them is
	denoted
	\begin{equation}
		t = \tau_v(s),
	\end{equation}
	where $\tau_v(s)$ is called the \emph{recorded emission time} and the
	difference
	\begin{equation}
		\delta\tau_v(s) = s - \tau_v(s),
	\end{equation}
	is called the \emph{time delay}.
\end{definition}

By construction, the following theorem holds.
\begin{theorem}\label{thm-gauge-inv}
	Given two gauge-equivalent lab-equipped spacetimes $(M,g,O,e^a_i)$ and
	$(M',g',O',e^{\prime a}_i)$, the corresponding time delays (keeping
	$s$ and $v^i$ fixed) are equal:
	\begin{equation}
		\delta\tau_v(s) = \delta\tau'_v(s).
	\end{equation}
\end{theorem}
In other words, the time delay (as well as as any function thereof, such
as the recorded emission time) constitutes a genuine (diffeomorphism-invariant)
physical observable on lab-equipped spacetimes. When the
context is clear, we will omit the explicit dependence of
$\delta\tau_v(s)$ on $v$ or $s$.

\section{Causal inequalities}\label{caus-ineq}
The time delay is interesting in more ways than just being an explicit
example of a physical observable sensitive to the ambient gravitational
field. The
experimental protocol defining it can be thought of as designed to test
the impossibility of superluminal signal propagation and the geodesic
character of the lab worldline. In particular, under quite generic
assumptions related to the Lorentzian character of spacetime geometry,
the time delay satisfies inequalities that would be violated if the
above-mentioned tests were to fail. A careful examination of a
mathematical model of such tests for classical spacetimes can serve as a
benchmark to understand possible outcomes of such tests in any proposed
theory of quantum spacetime.

\subsection{Maximality of light speed}
In curved spacetimes, it is impossible to objectively compare the speed
of light at different spacetime events. Due to general covariance, any
experiment to measure the \emph{local} speed of light, calibrated at
spacetime event $x$ to return the value $1$, will return the same value
$1$ at any other spacetime event $y$, provided it was parallel-transported
there. Therefore any such experiment must perform
measurements in finite regions of spacetime.

In the case of the time delay experiment, since light is used to send
the signals, the local speed of signal propagation cannot by definition
exceed that of light. However, to take into account possible nontrivial
global geometry, we adopt the following definition of (apparent) superluminal signal
propagation. Consider a pair of emission-reception times $(\tau,s)$. If
there exists another pair $(\tau',s')$ such that $\tau'<\tau$ (signal emitted
later than $t'$) and $s'>s$ (signal arrived earlier than $s'$), then the
later signal must have travelled \emph{superluminally}, cf.~Fig.~\ref{triang-geom-caus}. In classical
Lorentzian geometries without closed causal curves 
we can prove that this never happens.

\begin{figure}
\includegraphics{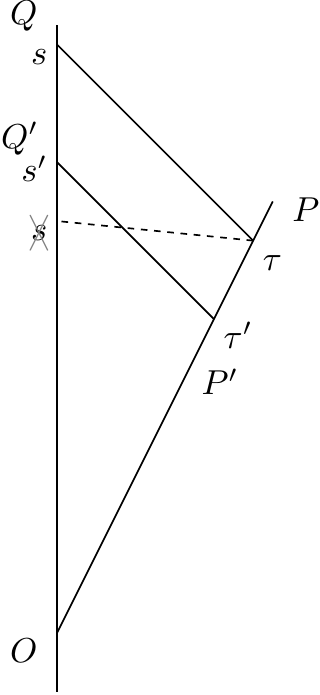}
\caption{%
Illustration of the conclusion of Theorem~\ref{light-bound}.
Successive emission times $\tau < \tau$ imply successive
reception times $s < s'$. The dashed line represents a case, ruled out
by the theorem, where the signal might appear superluminal.
}
\label{triang-geom-caus}
\end{figure}

\begin{theorem}\label{light-bound}
	In a lab-equipped spacetime $(M,g,O,e^a_i)$ we have the following
	implication between inequalities satisfied by pairs $(\tau(s),s)$ and
	$(\tau(s'),s')$ of emission-reception times:
	\begin{equation}
		\tau(s') < \tau(s) ~ \implies ~ s' < s .
	\end{equation}
	In particular, when $\tau(s)$ is smooth, we have $\frac{\d}{\d
	s}\tau(s) > 0$.
\end{theorem}
Essentially, this theorem says that a signal that is emitted later,
with respect to the probe, also arrives later, with respect to the lab.
\begin{proof}
Let the two signals be emitted from chronologically successive points
$P'$ and $P$, and received at points $Q'$ and $Q$, respectively. Since
$P'$ and $P$ are part of the same worldline, $P$ clearly belongs to the
set of all points that can be reached from $P'$ by future-directed
timelike curves, $P\in I^+(P')$.  The points $Q$ and $Q'$ are also
connected by a timelike curve, though we do not assume in which
precedence order, therefore $Q$ must belong to either $I^+(Q')$ or
$I^-(Q')$.  At the same time, by the definition of $Q$, it can be
reached by a piecewise smooth, non-spacelike, future-directed curve
$P'PQ$. Since $P'PQ$ is obviously not a null geodesic, Prop.~4.5.10
of~\cite{hawking-ellis} implies that $Q$ and $P'$ can be joined by a
(future-directed) timelike curve, $Q\in I^+(P')\subset
\operatorname{int} J^+(P')$.
By definition, $Q'$ is reached from $P'$ by a future-directed null
geodesic, such that there is no later point $P''$ on the probe
worldline with the same property. This implies that $Q'\in \del
J^+(P')$. Otherwise, $Q'\in I^+(P')$, hence $P'\in I^-(Q')$, hence any
point $P''\in \del J^-(Q')$ that is also on the probe worldline
violates the preceding hypothesis.
But all the timelike curves from $\del J^+(P')$ to $\operatorname{int}
J^+(P')\ni Q$ can only be future directed,
of which the one reaching $Q$ from $Q'$ is a special case, hence $Q\in
I^+(Q')$. This shows that $Q'$ chronologically precedes $Q$ or $s' < s$.
\end{proof}

\subsection{Geodesic extremality}
The twin ``paradox'' is a well-known phenomenon in special relativity:
the proper time between two timelike separated events is maximized by a
straight line (inertial motion). Its generalization to curved spacetime
is generally true only locally: a timelike geodesic maximizes proper
time among causal curves close to it (provided it has no conjugate
points). 
Under some conditions on the spacetime or under some extra restrictions
on the class of allowed causal curves, geodesic extremality can also
hold globally. This includes the special geometry of the time delay
experiment. As we shall see below, the time delay observable is also
sensitive to some violations of the geodesic extremality. Such
violations mimic a breakdown of the equivalence principle (objects no
longer fall on geodesics in the absence of external forces).

\begin{figure}
\includegraphics{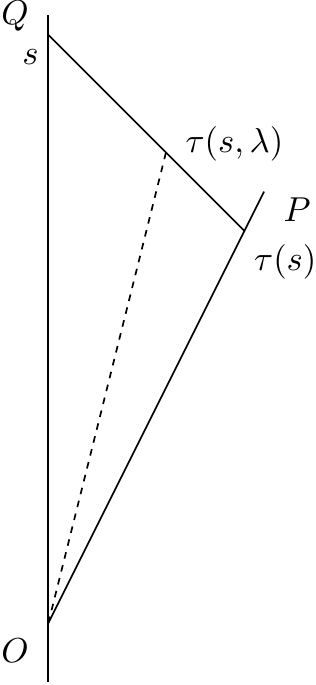}
\caption{%
Illustration of the proof of Theorem~\ref{twin-bound}. The auxiliary dashed
curve interpolates between $OQ$ and $OP$, as $\lambda$ varies from $0$
to $1$. As it does so, its proper time length $T(\lambda)$ is shown to
decrease monotonically, thus implying $\tau(s) < s$.
}
\label{triang-geom-twin}
\end{figure}

\begin{theorem}\label{twin-bound} 
	In a lab-equipped spacetime $(M,g,O,e^a_i)$ (where the lab and probe
	worldlines are smoothly deformable into each other in a sense to
	be precised in the proof)
	a pair $(\tau(s),s)$ of emission-reception times satisfies the
	inequality
	\begin{equation}
		\tau(s) < s \quad \text{or equivalently} \quad \delta\tau(s) < 0 .
	\end{equation}
\end{theorem}
\begin{proof}
The basic strategy of the proof is to construct a one-parameter family of
piecewise geodesic curves that interpolate between the lab worldline
$OQ$ and the probe-signal worldline $OPQ$, while their proper time
lengths decrease monotonically, cf.~Fig.~\ref{triang-geom-twin}. The
existence of the specific interpolation constructed below is the extra
technical hypothesis alluded to in the statement of the theorem.

Suppose that the $PQ$ geodesic is affinely parametrized as $Z(\lambda)$,
where $Z(0)=Q$ and $Z(1)=P$. Denote also $P_\lambda=Z(\lambda)$.  Then
the family $OP_\lambda Q$ clearly interpolates between $OQ$ and $OPQ$,
where $OP_\lambda$ is a timelike geodesic connecting these points and
$P_\lambda Q$ is the segment $Z([0,\lambda])$. Since the $P_\lambda Q$
segment is null, only the $OP_\lambda$ segment contributes to the proper
time $T(\lambda)$ along $OP_\lambda Q$. Since $T(0)=s$ and
$T(1)=\tau(s)$, the proof is concluded as soon as we show that
$\frac{\d}{\d\lambda} T(\lambda)<0$, which we do below.

We adapt the calculation of the first variation of the proper time
length of a piecewise geodesic curve from Prop.~4.5.4
of~\cite{hawking-ellis}. Let $Y(t,\lambda)$ denote the geodesic family
$OP_\lambda$, parametrized such that $Y(0,\lambda)=O$ and $Y(1,\lambda)
= Z(\lambda)$. $Y(t,\lambda)$ is assumed to be smooth in both arguments
by the smooth deformability hypothesis of the theorem.  Denote
$\dot{Y}^a = \frac{\del}{\del t} Y(t,\lambda)$, $Y^{\prime a} =
\frac{\del}{\del\lambda} Y(t,\lambda)$ and $f = [-g_{ab}(Y(t,\lambda))
\dot{Y}^a \dot{Y}^b]^{1/2}$. Also, for the purposes of the calculation
below, pick a coordinate chart $x^\alpha$ and replace the Latin abstract
tensor indices by Greek coordinate indices.

\begin{align}
	f' &= -f^{-1} \left[ g_{\alpha\beta} \dot{Y}^{\prime\alpha} \dot{Y}^\beta
			+ \frac{1}{2} g_{\alpha\beta,\gamma} Y^{\prime\gamma}
			\dot{Y}^\alpha \dot{Y}^\beta \right] \\
		&= -\dot{Y}^{\prime\gamma} g_{\gamma\beta} \frac{\dot{Y}^\beta}{f}
			- Y^{\prime\gamma} \frac{1}{2} g_{\alpha\beta,\gamma}
				\frac{\dot{Y}^\alpha\dot{Y}^\beta}{f} \\
		&= -\frac{\del}{\del t} \left( g_{\alpha\beta} Y^{\prime\alpha}
					\frac{\dot{Y}^\beta}{f} \right) \\
\notag
		&\quad {}
			+ f Y^{\prime\gamma} \left[
				\frac{1}{f} \frac{\del}{\del t}
					\left(g_{\gamma\beta}\frac{\dot{Y}^\beta}{f}\right)
				- \frac{1}{2}g_{\alpha\beta,\gamma} \frac{\dot{Y}^\alpha
					\dot{Y}^\beta}{f^2}
			\right] \\
		&= -\frac{\del}{\del t} \left( g_{ab} Y^{\prime a}
					\frac{\dot{Y}^b}{f} \right)  .
\end{align}
Note that the bracketed term vanished because it is precisely the
geodesic condition (Eq.~87.3a of~\cite{landlif2}) and $Y(t,\lambda)$ is a geodesic for fixed
$\lambda$. At $t=0$, we have $Y^{\prime a} = 0$, while at $t=1$, we have
$Y^{\prime a} = Z^{\prime a}$, which is a past-directed null vector.
\begin{align}
	T'(\lambda)
		&= \frac{\del}{\del\lambda} \int_0^1\d{t}\,[-g_{ab}\dot{Y}^a\dot{Y}^b]^{1/2}
			= \int_0^1 \d{t}\,f' \\
		&= -\int_0^1 \d{t}\, \frac{\del}{\del t}
			\left( g_{ab} Y^{\prime a} \frac{\dot{Y}^b}{f} \right) \\
		&= - g_{ab} Z^{\prime a} \frac{\dot{Y}^b}{f} < 0  .
\end{align}
The latter inequality follows because $\dot{Y}^a/f$ is a future-directed
timelike unit vector and $Z^{\prime a}$ is a past-directed null vector,
hence their inner product is positive. Armed with this inequality, it
immediately follows that
\begin{equation}
	s-\tau(s) = T(0)-T(1) = -\int_0^1\d\lambda\,T'(\lambda) > 0,
\end{equation}
which completes the proof.
\end{proof}

\section{Explicit calculation in classical linearized gravity}\label{class-calc}
The time delay observable, while well-defined from its description in
the preceding sections, has so far been defined only implicitly.
Unfortunately, it would be very difficult to obtain an explicit
expression for it, except in highly symmetric spacetimes, where
the required geodesics can be computed explicitly. In particular, in
Minkowski space, as is done below, it can be computed by elementary
means. Fortunately, for small perturbations of Minkowski space, an
explicit expression for the time delay can be found at linear order.
Such an expression would be especially needed for the calculation of
quantum averages and fluctuations, as sketched in Sec.~\ref{quant-calc}.

The calculations are carried out in the tetrad formalism. While
linearized gravity calculations are usually carried out in the more
familiar metric variables, there are a few reasons to consider tetrads.
Using tetrads opens the door to a kind of improved perturbation theory,
where the metric keeps its Lorentzian signature at every step of the
approximation. This line of investigation, as briefly brought up in
Sec.~\ref{discussion}, will be pursued elsewhere.  Another advantage of
tetrads is that they are needed in the standard way of formulating
fermions on curved spacetime.

First, we explicitly compute the time delay in Min\-kow\-ski space and check
the causal inequalities. Then, using the results of the perturbative
solution of the geodesic and parallel transport equations of the
\hyperref[pertsol]{Appendix}, we compute the explicit expression for the time
delay at linear order in the deviation from Minkowski space.

\subsection{Minkowski space}
Consider Minkowski space $(M=\mathbb{R}^4,\eta,0,\hat{x}^a_i)$, with
$\eta=\mathrm{diag}(-1,1,1,1)$, as a lab-equipped spacetime. Without
loss of generality, we can take an arbitrary inertial coordinates $x^i$
on $(M,\eta)$ and use their origin $0$ as the synchronization point and
the vectors $\hat{x}^a_i = (\del/\del x^i)^a$ as the reference tetrad.
The dual tetrad is $\hat{x}^i_a = (\d x^i)_a$ and satisfies the
identities $\hat{x}^a_i \hat{x}_a^j = \delta_i^j$ and $\hat{x}^a_i
\hat{x}_b^i = \delta^b_a$. The Minkowski metric is $\eta_{ab} =
\eta_{ij} \hat{x}^i_a \hat{x}^j_b$.

The lab and probe worldlines are parametrized, respectively, as $x^i(s)=s
u^i$ and $x^i(t) = t v^i$, $i=0,1,2,3$, where $u^i=(1,0,0,0)$. Suppose that
the relative speed of the two timelike vectors $u^a = u^i \hat{x}^a_i$
and $v^a = v^i \hat{x}^a_i$ is given by the positive hyperbolic rapidity
$\theta$, $v_\mathrm{rel} = \tanh\theta$, then we have $u\cdot v =
\eta_{ab} u^a v^b = \eta_{ij} u^i v^j = -\cosh\theta$. The values
of $s$ and $t$ which may be connected by light signals are constrained
by
\begin{align}
	\eta_{ij} (s u^i - t v^i) (s u^j - t v^j) &= 0 \\
	-s^2 - 2 st (u\cdot v) - t^2 &= 0 \\
	s^2 - 2 st \cosh\theta + t^2 &= 0 \\
	(s e^\theta - t)(se^{-\theta}-t) &= 0 .
\end{align}
The retarded solution is then
\begin{equation}\label{emmtime-mink}
	\tau_\mathrm{cl}(s) = t = s e^{-\theta} .
\end{equation}
The subscript stands for ``classical,'' as it will serve in Sec.~\ref{quant-calc}
as the classical background expectation for quantum fluctuations.  This
expression clearly satisfies the causal inequalities obtained in the
previous section:
\begin{align}
	\tau_\mathrm{cl}(s) &= s e^{-\theta} < s, \\
	\frac{\d}{\d s} \tau_\mathrm{cl}(s) &= e^{-\theta} > 0.
\end{align}
Note that $\theta > 0$ since the probe is moving away from the lab.
The null vector connecting the emission and absorption points is
\begin{equation}
	w^i = su^i - tv^i = s (u^i - e^{-\theta} v^i).
\end{equation}
Another useful identity is
\begin{equation}
	v^i = e^\theta(u^i-w^i/s).
\end{equation}

\subsection{Approximately Minkowski space}
\subsubsection{Tetrad formalism}\label{tetform}
Consider another lab-equipped spacetime $(M,g,0,\hat{e}^a_i)$, where
we have kept the same underlying manifold $M$ and synchronization point
$O=0$ as in Minkowski space. On the other hand, we express the new
metric as $g_{ab} = \eta_{ij} e^i_a e^j_b$, where $e^i_a$ and $e^a_i$ is
a new dual pair of orthonormal tetrads,
\begin{equation}
	e^a_i e^b_j g_{ab} = \eta_{ij},
	~ e^a_i e_a^j = \delta_i^j,
	~ e^a_i e_b^i = \delta_a^b.
\end{equation}
Using the Minkowski tetrad $\hat{x}^a_i$ on $M$ as a reference, any other one
can be obtained by a local general linear transformation
\begin{equation}\label{tetref}
	e^a_i = \bar{T}^{i'}_i \hat{x}^a_{i'},
	~ e^i_a = T^i_{i'} \hat{x}^{i'}_a,
\end{equation}
where $T$ and $\bar{T}$ are spacetime-dependent invertible matrices, such that
$\bar{T}=T^{-1}$. Similarly, any lab frame $\hat{e}^a_i$ can
be obtained by another general linear transformation at $O$,
\begin{equation}
	\hat{e}^a_i = (T_O)^{i'}_i \hat{x}^a_{i'}.
\end{equation}
The possible discrepancy between the lab frame and the spacetime tetrad
at $O$ is
\begin{equation}
	\hat{e}^a_i = L^{i'}_i e^a_{i'} ,
	\quad
	L^{i'}_i = (T_O)^{i'}_j T^j_i ,
\end{equation}
where $L$ is clearly a Lorentz transformation, $L^{i'}_i L^{j'}_{j}
\eta_{i'j'} = \eta_{ij}$.

If this new lab-equipped spacetime is approximately Minkowski, then
both $T^i_{i'}$ and $L^{i'}_i$ must be close to the identity matrix.
This is conveniently expressed by first parametrizing them as
$T=\exp(h)$ and $L=\exp(h_O)$, and then requiring that $h$ and $h_O$ are
close to $0$. The smallness requirement aside, $h$ and $h_O$ could be,
respectively, an arbitrary matrix and an arbitrary skew-adjoint matrix,
$\eta_{ik} (h_O)^k_j = -\eta_{kj} (h_O)^k_i$. Then the metric is
\begin{align}
	g_{ab} &= \eta_{ij} e^a_i e^b_j
			= \eta_{ij} T^{i}_{i'} T^{j}_{j'} \hat{x}_a^{i'} \hat{x}_b^{j'} \\
		&= \eta_{ab} + (\eta_{i'j} h^i_{i'} + \eta_{ij'} h^j_{j'})
			\hat{x}_a^{i'} \hat{x}_b^{j'} + \O(h^2) \\
		&= \eta_{ab} + \tilde{h}_{ab}.
\end{align}
The last two equations describe the relationship between the deviations
$h^i_j$ and $\tilde{h}_{ab}$ from Minkowski space, in the tetrad and
metric formalisms respectively,
\begin{equation}
	\tilde{h}_{ab} = 2h_{(ij)} \hat{x}^i_a \hat{x}^j_b + \O(h^2),
	\quad \text{where} \quad
	h_{ij} = \eta_{i j'} h^{j'}_j.
\end{equation}

A worldline $\gamma(t)$ is described by its coordinates
$\gamma^i(t)=x^i(\gamma(t))$. Its tangent vector is denoted
$\dot{\gamma}(t)^a$.  Knowledge of the tangent vector allows one to
recover the curve as follows
\begin{equation}
	\int_{t_1}^{t_2} \d{t}\, \dot{\gamma}^a(t) (\d{x}^i)_a
	= \int_{\gamma(t_1)}^{\gamma(t_2)}\d{x}^i
	= \gamma^i(t_2) - \gamma^i(t_1) .
\end{equation}
For
convenience, all curves are affinely parametrized from $0$ to $1$.
Thus, the length of a timelike geodesic is equal to the length of
its initial tangent vector.

A geodesic $\gamma(t)$ is completely specified by its point of origin
$\gamma(0)$ and its initial tangent vector $\dot{\gamma}^a(0)$, while a
$\gamma$-parallel-transported vector $v^a(t)$ is specified by its
initial value $v^a(0)$ at $\gamma(0)$. Again, for convenience in further
calculations, all such initial data are specified with reference to some
given curve $\beta$, with $\beta(0)=O$. Namely, the point of origin is
$\gamma(0)=\beta(1)$, the initial tangent vector $\dot{\gamma}^a(0)$ is
the $\beta$-parallel-transported image of a vector $\dot{\gamma}^a_O =
\dot{\gamma}^i_O \hat{e}^a_i$, and the initial value $v^a(0)$ is the
$\beta$-parallel-transported image of a vector $v^a_O = v^i_O
\hat{e}^a_i$ (cf.~Fig.~\ref{geodesic-geom}).
The geodesic and parallel transport equations are written down and
solved to order $\O(h)$ in the \hyperref[pertsol]{Appendix}.

\subsubsection{Geodesic triangle construction}\label{triang-lin}
All curves considered in this section are perturbations of piecewise
linear paths, which are piecewise geodesic in Minkowski space. In
particular, at zeroth order in $h$, the sides of the geodesic triangle
formed by the worldlines of the lab, the probe, and the signal form an
ordered sequence of spacetime segments $(V,W,U)$, as illustrated in
Fig.~\ref{triang-geom-H}. Namely, $V$ stretches from $O$ to $P$, $W$ stretches from
$P$ to $Q$, and $U$ stretches from $Q$ back to $O$. Using the
convention of the last paragraph of the preceding section, each of the
$(V,W,U)$ segments can be specified as starting from the end point of
the preceding one (note that the order corresponds to counterclockwise
starting from $O$ in Fig.~\ref{triang-geom-H}) with the respective tangent vectors
$(tv^a,w^a,-su^a)$. Because Minkowski space is flat, it is clear that the
segments $VWU$ form a closed triangle by virtue of their tangent vectors
adding up to zero.

In approximately Minkowski space, we wish to describe a perturbed
version of the above construction. Namely, a sequence of geodesic
segments $(\tilde{V},\tilde{W},\tilde{U})$, connected from end to end,
with the respective images $(\tilde{t}\tilde{v}^a, \tilde{w}^a,
-\tilde{s}\tilde{u}^a)$ of their initial tangent vectors parallel-transported
to $O$. We take $\tilde{u}^a$ and $\tilde{v}^a$ to be unit
vectors, hence $\tilde{s}$ and $\tilde{t}$ are the proper time lengths
of the corresponding segments. To be consistent with the experimental
protocol described in Secs.~\ref{op-def} and~\ref{class-model}, we must take $\tilde{s} = s$
and $\tilde{v} = v^i \hat{e}^a_i$, require that $\tilde{w}^a$ is null,
require that the geodesic triangle closes (the end point of $U$ is in
fact $O$), and finally that the tangent to $U$ at $O$ is
$-su^i\hat{e}^a_i$ (which is also the parallel-transported image along
the $VWU$ triangle, in other words a holonomy image, of
$-\tilde{s}\tilde{u}^a$):
\begin{align}
	\tilde{s} &= s , \\
	\tilde{t} &= e^{\tilde{r}} t , \\
	\tilde{v}^a &= v^i \hat{e}^a_i = e^a_i \exp(h_O)^i_j v^j , \\
	\tilde{u}^a &= e^a_i [\exp(p_U)\exp(p_W)\exp(p_V)]^i_j (h_O)^j_k u^k , \\
	\tilde{w}^a &= e^a_i \exp(\tilde{q})^i_j w^j ,
\end{align}
where we have used the notation $\exp(p_\gamma)$ for the parallel
transport operator along $\gamma$, Eq.~\eqref{expp-def}, while
$\tilde{r}$ is a scalar and $\exp(\tilde{q})$ a Lorentz transformation
($\eta_{ik}\tilde{q}^k_j=-\eta_{kj}\tilde{q}^k_i$), both yet to be
determined. Note that $\tilde{q}^i_j$ does not parametrize $\tilde{w}^a$
uniquely, as $\exp(q)$ could always be premultiplied by another Lorentz
transformation fixing $w^k$, but it does contain three non-arbitrary
parameters. The only condition left to be satisfied is the closure of
the $VWU$ triangle (equating the end point of $U$ with $O$), which
provides four equations. These four equations can be used to solve for
the remaining undetermined parameters, one in $\tilde{r}$ and three in
$\tilde{q}^i_j$. Since we are working at linear order, we only need the
leading terms in the expansion of these unknowns
\begin{equation}
	\tilde{q}^i_j = q^i_j + \O(h^2) ,
	\quad
	\tilde{r} = r + \O(h^2) .
\end{equation}

Using the perturbative solution of the geodesic and parallel transport
equations obtained in the Appendix (Eqs.~\eqref{prlt-sol} and~\eqref{geod-sol}), at linear
order, the triangle closure condition can be written out explicitly as
\begin{align}
	0
	&= (tv^i + rtv^i + J^i_{V,\varnothing}) \\
\notag & \quad {}
		+ (w^i + q^i_j w^j + J^i_{W,V}) \\
\notag & \quad {}
		+ (-su^i + H^i_j su^j + J^i_{U,VW}) \\
	&= rtv^i + q^i_j w^j + H^i_j su^j + J^i,
\end{align}
where, using the notation of Eqs.~\eqref{H-grp} and~\eqref{J-grp}, we have defined
\begin{align}
\label{H-def}
	\eta_{ik} H^k_j = H_{ij} &= (H_{(V,W,U)})_{ij} , \\
\label{J-def}
	\eta_{ik} J^k = J_i &= (J_{(V,W,U)})_i .
\end{align}
The expression in parentheses vanishes due to the closure of the zeroth-order
geodesic triangle. Also, contracting the closure condition with
$w^i$ makes the term with $q$ vanish (due to its antisymmetry). The
solution for $r$ is then
\begin{equation}\label{r-def}
	r = -\frac{w^i J_i + w^i H_{ij} su^j}{\tau_\mathrm{cl}(s) v\cdot w} .
\end{equation}

The detailed structure of the defining expression for $r$ in
Eq.~\eqref{r-def} can be deduced from the structure of the expressions
for the $H$ and $J$ terms, given explicitly in Eqs.~\eqref{H-details}
and~\eqref{J-details}.  It can be described as follows. Both $H$ and $J$
consist of a sum of terms associated to the segments of the $VWU$
triangle. A term associated to segment $X$ consists of a tensor, built
up from the vectors $u^i$, $v^i$ and $w^i$, contracted with a (possibly
iterated) line integral over $X$, where the integrand consists of the
perturbation $h_{ij}$, possibly with several derivatives applied to it.
Schematically, this structure can be expressed as
\begin{equation}
	r \sim \sum_X r_{X,m,k} \int_X^{(m)}\d{t} \nabla^k h ,
\end{equation}
where all tensor indices have are suppressed and iterated integrals are
represented using the notation from
Eqs.~\eqref{apdx-not1}--\eqref{apdx-not4}. There is at most one
derivative ($k\le 1$) and integration over a spacetime segment is
iterated at most twice ($m\le 1$). In a bit more detail, though leaving
the tensor contractions aside, the structure of the $H$ and $J$ terms
can be expressed as follows
\begin{align}
\label{H-struct}
	H &\sim \sum_{X=V,U,W} (\sarr)_X , \\
\label{J-struct}
	J &\sim \sum_{X=V,U,W} \left[ (\darr)_X + \sum_{Y<X} (\sarr)_Y \right] , \\
\label{sarr-struct}
	(\sarr)_X &\sim \int_X \nabla h , \\
\label{darr-struct}
	(\darr)_X &\sim \int_X h + \int_X^{(1)} \nabla h,
\end{align}
where the order between the segments is counterclockwise starting from
$O$, as in Figs.~\ref{triang-geom-H} and~\ref{triang-geom-J}. The geometry of the various terms is illustrated
in Figs.~\ref{triang-geom-H} and~\ref{triang-geom-J}. This information
is used in Sec.~\ref{dim-anl}.

\begin{figure}
\includegraphics{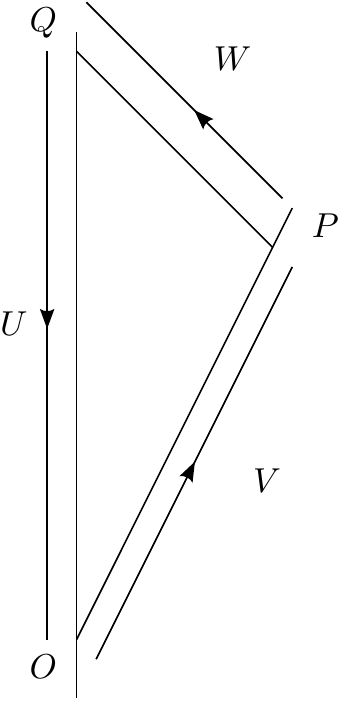}
\caption{%
Schematic structure of the $H$-term in $r$, Eq.~\eqref{H-struct}.
Notation follows Eqs.~\eqref{H-struct}--\eqref{darr-struct}.
}
\label{triang-geom-H}
\end{figure}

\begin{figure}
\includegraphics{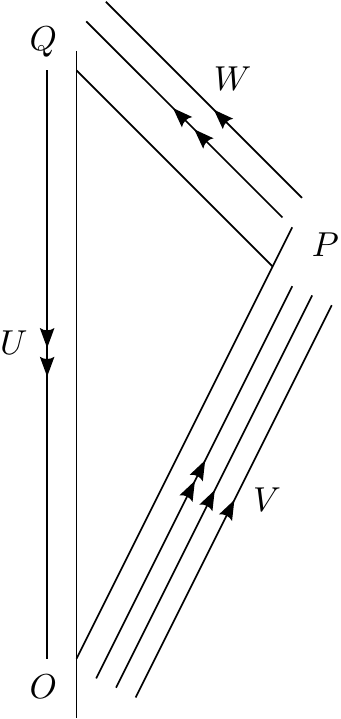}
\caption{%
Schematic structure of the $J$-term in $r$, Eq.~\eqref{J-struct}.
Notation follows Eqs.~\eqref{H-struct}--\eqref{darr-struct}.
}
\label{triang-geom-J}
\end{figure}

\subsubsection{Time delay and gauge invariance}\label{gauge-inv}
As proven in Theorem~\ref{thm-gauge-inv}, the time delay is a gauge-invariant
observable. From the formula 
\begin{equation}\label{emmtime-lin}
	\tau(s) = \tilde{t} = e^{\tilde{r}} \tau_\mathrm{cl}(s)
	= \tau_\mathrm{cl}(s) + r\tau_\mathrm{cl}(s) + \O(h^2),
\end{equation}
that relates the linearized gravity correction $r\tau(s)$,
Eq.~\eqref{r-def}, to the Minkowski space result $\tau_\mathrm{cl}(s)$,
Eq.~\eqref{emmtime-mink}, it is obvious that $r$ should be invariant
under linearized gauge transformations. This can be checked explicitly
using the gauge transformation formulas, Eqs.~\eqref{H-trans}
and~\eqref{J-trans}, for the terms making up $H_{ij}$ and $J_i$.
As a consequence, which is given at the bottom of the
\hyperref[pertsol]{Appendix}, the closure of the $VWU$ triangle in
Minkowski space implies the individual gauge invariances of both
$H_{ij}$ and $J_i$, and hence of $r$.

The last remark deserves some emphasis. There have been many attempts to
try to achieve some sort of explicit and complete classification of gauge-invariant observables
of general relativity~\cite{bergmann-obsv, rovelli-obsv, dittrich-obsv,
lp-obsv, pss-obsv}. So far, no such complete classification is known.
Even in the case of a partial classification, such lists of gauge-invariant
observables are often obtained without direct physical interpretation.
The strategy of this paper has been different. The idea was to first
establish an operational definition of an observable, in terms of the
thought experiment described in Sec.~\ref{op-def}, second to
establish a mathematical model thereof, which would naturally be gauge-invariant
though perhaps only defined implicitly, and third to use an
approximation method (linear-order perturbation theory, in this case) to
obtain an explicit expression for the observable. The result of this
strategy is an explicit (linearly) gauge-invariant expression for an
observable $r$ and a physical interpretation of it as an approximation
to the outcome of a clearly described thought experiment. It is of
course highly likely that an exhaustive classification of gauge-invariant
observables, for the simpler problem of linearized gravity, would have identified
explicit expressions like $H_{ij}$ and $J_i$, but it is at the same time
highly doubtful that they would be accompanied by the clear physical
interpretation we have managed to associated to their particular
combination in~\eqref{r-def}.

It is also worth noting that the works of Ford \emph{et
al.}~\cite{ford-lightcone, ford-top, ford-focus, ford-angle} and
Roura and Arteaga~[13] 
worked in a particular gauge and with more restricted experiment
geometries. Thus they did not obtain the same general gauge-invariant
expressions that we have derived here. However, similar expressions,
expanded even to quadratic order, were obtained in the work of Tsamis
and Woodard~\cite{tsamis-woodard}.

\section{Sketch of quantum mathematical model}\label{quant-model}
Ideally, to be able to theoretically describe quantum effects, the
thought experiment protocol described in Sec.~\ref{op-def} should be
translated into a mathematical model within a quantum theory that
encompasses both the gravitational field and the experimental apparatus
described in the protocol. A na\"ive attempt to do this is obstructed by
several difficulties: (a) the lack of a uniformly accepted (or at the
very least sufficiently general) quantum theory of gravity, (b) the
identification of a time observable in quantum mechanics, and (c) the
difficulties in modeling measurements in quantum mechanics.
Fortunately, we can propose pragmatic solutions to each of these
problems, as discussed below.

\subsection{Quantum linearized gravity}
While it is true that there is no uniformly accepted theory of quantum
gravity, there are some common standards that are expected to be met by
the final version of any proposal.
One such routine benchmark is the
ability to reproduce classical general relativity in the appropriate
limit. It is worth noting that under very general circumstances (in the
absence of strong curvatures), the dynamics of the gravitational field in
general relativity can be very closely approximated by the dynamics of
linearized gravity, also known as the theory of (linear) gravitational
waves. Our experience to date overwhelmingly demonstrates that the
quantum theory of any field whose dynamics may be approximated by a linear
theory, be it a ``fundamental'' field as in elementary particle physics
or an ``effective'' field as in condensed matter theory, is well-approximated
by the Fock quantization of the approximate linear theory.
By inductive reasoning, we presume that any proposed theory of quantum
gravity should also be benchmarked by its ability to reproduce quantum
linearized gravity. Therefore, pragmatically, we restrict ourselves to
the Fock quantization of the linearized gravity field on Minkowski space
as the approximate quantum theory of gravity for the purposes of the
mathematical model of the time delay observable.

\subsection{Time in quantum mechanics}
It is often repeated physics lore that there is no observable in quantum
mechanics corresponding to time, which naturally leads one to wonder
whether it is even possible to model time measurements in quantum
mechanics. This argument is originally due to Pauli (p.63, footnote 2
of~\cite{pauli-time}).  Fortunately, when precisely stated, it is much less
restrictive than one is first lead to believe~\cite{hilg-time,
muga-time, olkh-time, bfh-time}. The crux of this argument is a
contradiction that stems from the following hypotheses. Suppose we have
a quantum mechanical system with Hamiltonian $\hat{H}$, whose spectrum
is bounded from below, and an operator observable $\hat{T}$, whose
commutation relation with $\hat{H}$ is precisely of the form
$[\hat{T},\hat{H}]=i\hbar$, as would be appropriate for a ``time
observable'' $\hat{T}$ (together with appropriate continuity and
functional analytical conditions). Then, an appeal to the
Stone-von~Neumann uniqueness theorem (Theorem~VIII.14 in~\cite{rs1})
establishes a contradiction, as, according to the theorem, both
$\hat{T}$ and $\hat{H}$ must have continuous unbounded spectra. Thus,
there cannot exist such an observable $\hat{T}$ corresponding to time.
However, there are at least two physically reasonable ways to circumvent
this conclusion. One is to drop the hypothesis that $\hat{H}$ is bounded
from below. While this requirement is important for the global, long-term
stability of physical systems, its not necessary in some
approximate descriptions meant to describe the dynamics of some system
for bounded time intervals. Two common examples are a particle in a
linear potential and a harmonic oscillator with an inverted potential.
The other is to relax the commutation relation condition to
$[\hat{T},\hat{H}]\approx i\hbar$, where the correction terms that
restore equality may be higher order in $\hbar$ or may be small in
another way when restricted to a physically relevant subspace of
possible states.  An example is a particle on a circle, whose dynamics
dictate uniform motion, so that its position can serve as an approximate
``cyclic time'' observable, like the position of the hand of an analog
clock. Many more examples are discussed in~\cite{hilg-time, muga-time,
olkh-time} and the references therein.

\subsection{Modeling quantum measurements}
\subsubsection{Classical vs quantum measurements}
The remaining obstacle is overcome by constructing a fairly explicit,
though still rough, model of a measurement, where the system of interest
(gravitational field, lab, probe, signal), the measurement devices
(proper time clocks) and recording devices are all taken into account.
The details of this setup are described below, following some of the
ideas of~\cite{pw, gppt-montevideo} on the use of physical clocks in
quantum systems. The conclusion can be formulated as follows. After the
reception and emission times of a signal have been measured by the lab
and individually stored, the states and the dynamics of the storage
devices stabilize and decouple from the rest of the system, as well as
from each other, in the asymptotic future. Then, in the asymptotic
future, the corresponding ``readout'' observables $\hat{S}$ (recorded
reception time) and $\hat{T}$ (recorded emission time) commute and thus
define a joint (classical) probability distribution $\rho(\sigma,t) =
\langle \delta(\hat{S}-\sigma) \delta(\hat{T}-t) \rangle$, where the
expectation value is taken with respect to the (Heisenberg) state of the
total system, which we will refer to as the \emph{quantum gravitational
vacuum}. Mathematically, this probability distribution $\rho(\sigma,t)$
may be referred to as either the \emph{joint spectral density} of the
quantum gravitational vacuum with respect to the operators $\hat{S}$ and
$\hat{T}$. In more physical terms, $\rho(\sigma,t)$ is the
\emph{absolute value squared of the wave function} of the quantum
gravitational vacuum projected onto the variables $\sigma$ and $t$.

Recall that the main output the classical mathematical model of the
measurement of the time delay observable is the functional relation $t =
\tau(s)$, which can be seen as a special case of a joint probability
distribution $\rho_\mathrm{cl}(\sigma,t) =
\delta(\sigma-s)\delta(t-\tau(\sigma))$, where $s$ is the predetermined
time when the laboratory makes the measurements. This classical
probability distribution is so ``sharp'' because we take as a classical
state a definite configuration of the gravitational field.  More
generally, in the framework of classical statistical mechanics, we can
take any probability measure $d\rho(g)$ on the space of gauge
equivalence classes of the configurations of lab-equipped spacetimes.
The main output of the classical mathematical model of the time delay
measurement in this state is then the probability distribution
$\rho_{\mathrm{cl}}(\sigma,t) = \int d\rho(g)\, \delta(\sigma-s)
\delta(t-\tau_g(\sigma))$, where the dependence of the emission time
$\tau_g(s)$ on the equivalence class $g$ of lab-equipped spacetime
configurations is indicated through a subscript.  Thus, considering
quantum mechanics as an extension (or rather deformation) of classical
statistical mechanics, it is not surprising that the main output of a
quantum mathematical model of a measurement of the time delay observable
is the probability distribution $\rho(\sigma,t)$. Of course, being the
result of a quantum measurement, the distribution $\rho(\sigma,t)$
depends on more details of the measurement (such as the order in which
the measurements were carried out) than the classical distribution
$\rho_\mathrm{cl}(\sigma,t)$.

\subsubsection{Dynamical apparatus model}\label{dyn-appr}
The full system included in the model consists of the following
dynamical subsystems: the gravitational field $\hat{g}$, the lab and probe worldline
coordinates $\hat{y}_l$ and $\hat{y}_p$, the lab and probe proper time clocks
$\hat{\tau}_l$ and $\hat{\tau}_p$, the time registers $\hat{S}$ and
$\hat{T}$ in the lab, the
coordinates of the signal particles $\hat{z}$, and the time stamp
$\hat{\tilde\tau}$
carried by each signal particle. The spacetime is presumed to have a
fixed foliation by level sets of a time function $t$. The gravitational
field is taken to be completely gauge fixed, for instance using the
transverse, traceless, and $t$-compatible radiation
conditions (Sec.~4.4b of~\cite{wald}).
All worldlines can then be parametrized by $t$ as well. The dynamics of
the full system, describing its evolution with respect to time $t$, is
specified by a Hamiltonian
\begin{equation}
	\hat{H} = \hat{H}_{\mathrm{sub}} + \hat{H}_{\mathrm{geom}} + \hat{H}_{\mathrm{meas}} ,
\end{equation}
which is composed of $\hat{H}_{\mathrm{sub}}$ describing the independent
dynamics of the subsystems, of $\hat{H}_{\mathrm{geom}}$ describing the
necessary interactions or external interventions to effect the geometry
of the experimental setup, and of $\hat{H}_{\mathrm{meas}}$ describing the
coupling between the recording devices and the rest of the system during
the measurement.

Since we are mostly concerned here with a quantum model of measurement
of the clock readings, we will concentrate only on $\hat{H}_{\mathrm{meas}}$
and specify $\hat{H}_{\mathrm{sub}}$ and $\hat{H}_{\mathrm{geom}}$ mostly
verbally.

The dynamics of the gravitational field follow the appropriate gauge
fixed Hamiltonian, a term in $\hat{H}_{\mathrm{sub}}$ derived from the
Einstein-Hilbert action. The precise details of the implementation of
this idea are irrelevant for this discussion, as long as the
corresponding dynamics about the quantum gravitational vacuum can be
approximated by the dynamics of linearized gravity about the Fock
vacuum. This assumption is the basis of the calculation sketched in
Sec.~\ref{quant-calc}.

The worldlines of various particles are described by their spatial
coordinates as functions of the global time $t$, $\hat{y}^i_l(t)$,
$\hat{y}^i_p(t)$,
and $\hat{z}^i(t)$. The dynamics of these variables follow from the
appropriate terms in $\hat{H}_{\mathrm{sub}}$. The lab and probe worldlines
are timelike geodesics, with an appropriate term in $\hat{H}_{\mathrm{geom}}$
providing a kick to the probe at event $O$ to give it a fixed relative
velocity with respect to the lab. (The fact that $O$ lies on the lab
worldline can be used as one of the gauge-fixing conditions.) To imitate
the action of a continuously emitted signal field (like the
electromagnetic field), the multiplicity of signal particles are indexed
by a time $t'$ and a unit $3$-vector $\mathbf{n}$. The dynamics as
specified by terms in $\hat{H}_{\mathrm{sub}}$ and $\hat{H}_{\mathrm{geom}}$ should
be as follows. The worldline $\hat{z}^i_{(t',\mathbf{n})}(t)$ follows the
probe worldline until the time $t=t'$, after which point the worldline
of $\hat{z}^i_{(t',\mathbf{n})}(t)$ becomes null with direction determined by
$\mathbf{n}$. (This is a kind of eikonal approximation, which replaces a
massless field by a large collection of massless particles.) The initial
state of each of these particles is presumed to be of a localized
wave packet form, with negligible wave packet spread on time scales
comparable to the geometry of the experiment.

There are two potential problems in constructing a detailed quantum model
implementing the above requirements. While it is not difficult to write
down a classical version of such $\hat{H}_{\mathrm{sub}}+\hat{H}_{\mathrm{geom}}$,
the generalization to quantum mechanics is not unique, due to the usual
operator ordering ambiguities. The standard solution of this problem is
to parametrize these ambiguities and realize that different choices of
these parameters correspond to physically different models. Thus, the
fixation of these parameters must be part of the full specification of
the detailed model. Fortunately, these ordering ambiguities are
generically expected to be suppressed by powers of $\hbar$. Moreover, we
assume that their parametrization may be tuned to maximize the validity
of the approximations used in Sec.~\ref{quant-calc}. The second problem is
that coupling point particles to fields generically leads to singular
dynamics. (The singularities inherent in the naive interaction of a
classical point electron with its own electromagnetic field is a
classical example of this difficulty.) However, this issue can be dealt
with straightforwardly by spatial smearing of the particle-metric field
interaction terms. The spatial extent of the smearing becomes another
parameter whose value is to be chosen as to minimize the impact of the
smearing on the rest of the discussion. Alternatively, the interaction
term could be modified in a more sophisticated way, without introducing
non-local smearing, for instance along the lines suggested by the recent
work on classical point particles coupled to their self-force~\cite{poisson-lr} or by appealing to intrinsic quantum uncertainty of the
center of mass coordinates as in~\cite{glcm-spread}.

The state spaces for the clock and time register subsystems can be
presumed to be completely internal (i.e.,\ divorced from spacetime
coordinates) and thus can be subject to even further simplifications.
The time register subsystems $\hat{S}$, $\hat{T}$ and
$\hat{\tilde{\tau}}_{(t',\mathbf{n})}$ should be very stable, thus their
contribution to $\hat{H}_{\mathrm{sub}}$ should be approximately zero.
On the other hand, the clock variables $\hat{\tau}_l(t)$ and
$\hat{\tau}_p(t)$ should evolve approximately monotonically, with rates
set by their local proper time. This can be accomplished by a
contribution to $\hat{H}_{\mathrm{sub}}$ of the form $\dot{\hat{\tau}}_l
\hat{P}_{\tau_l} + \dot{\hat{\tau}}_p \hat{P}_{\tau_p}$, where
$\hat{P}_{\tau_l}$ and $\hat{P}_{\tau_p}$ are respectively canonically
conjugate to $\hat{\tau}_l$ and $\hat{\tau}_p$, while
$\dot{\hat{\tau}}_l$ and $\dot{\hat{\tau}}_p$ stand for the appropriate
expressions in terms of $\dot{\hat{y}}^i_l$, $\dot{\hat{y}}^i_p$ and
$\hat{g}$. Note that this choice of Hamiltonian circumvents Pauli's
impossibility argument by virtue of being unbounded from below.

Finally, $\hat{H}_{\mathrm{meas}}$ is chosen to implement the idea of
\emph{weak measurement}~\cite{bk-meas}. The idea of weak measurement can
be described as follows. Suppose there is a quantum variable $\hat{q}$
whose value we wish to measure and record in another variable $\hat{Q}$,
belonging to a recording device subsystem. Suppose that $\hat{P}$ is
canonically conjugate to $\hat{Q}$ and that $\hat{Q}$ suffers negligible
evolution on its own. Then the value of $\hat{Q}$ can be measured at any
convenient time after the weak measurement took place, thus allowing us
to infer (subject to quantum uncertainties) the value of $\hat{q}$ at
the time of measurement. The measurement itself can be modeled using the
interaction Hamiltonian $f_{\mathrm{trig}}(t) \hat{q} \hat{P}$, where
$f_{\mathrm{trig}}(t)$ is a \emph{trigger factor}, which is non-zero
only during the time interval when the measurement is supposed to take
place. The operators $\hat{q}(t)$ and $\hat{P}(t)$ are presumed to
commute at equal times, as they belong to independent subsystems, so
their ordering of the interaction Hamiltonian is unambiguous. If this
interval is of length $\Delta t$ and during it
$f_{\mathrm{trig}}(t)\approx f_0$ is approximately constant, we can see
that this interaction Hamiltonian effects the evolution
\begin{align}
\label{qafter}
	Q_{\mathrm{after}}
	&= e^{i\int\d{t}\, f_{\mathrm{trig}}(t) q P}
		Q_{\mathrm{before}} e^{-i\int\d{t}\, f_{\mathrm{trig}}(t) q P} \\
	&= Q_{\mathrm{before}} + \int\d{t}\, f_{\mathrm{trig}}(t) q \\
\label{qafter-approx}
	&\approx Q_{\mathrm{before}} + \Delta t f_0 q(t_0),
\end{align}
where the last approximation holds provided $t_0$ was part of the
measurement time interval and $\hat{q}(t)$ and $\hat{P}(t)$
evolved negligibly during it. In general, the trigger
$f_{\mathrm{trig}}$ need not be a
scalar, and may itself be a operator that commutes with both $\hat{q}(t)$ and
$\hat{P}(t)$. Also while each pair of factors commutes at equal times, in
general, they will not commute at unequal times (even with themselves).
Thus, the evolution effected by the interaction will involve the time-ordered
exponential of the interaction Hamiltonian and will look more
complicated than Eq.~\eqref{qafter}. However if the interaction Hamiltonian
can be considered as a small perturbation, then at linear order
$\hat{Q}_{\mathrm{after}}$ will look the same as Eq.~\eqref{qafter-approx}.

With the above discussion in mind, we set the measurement interaction
Hamiltonian to
\begin{multline}
	\hat{H}_{\mathrm{meas}} = \delta[\hat{\tau}_l-s]\dot{\hat{\tau}}_l \\
	{}\times 
	\left[ \hat{\tau}_l \hat{P}_S + \int\d{t'}\d{\mathbf{n}}\,
	\delta^3[\hat{z}_{(t',\mathbf{n})}-\hat{y}_l] \hat{\tilde{\tau}}_{(t',\mathbf{n})} \hat{P}_T \right].
\end{multline}
Note that the factors in each product commute at equal times (recall
that $\dot{\hat{\tau}}_l$ is \emph{not} canonically conjugate to $\hat{\tau}_l$), so
their ordering in $\hat{H}_{\mathrm{meas}}$ is unambiguous. The extra factor
of $\dot{\hat{\tau}}_l(t)$ is there to ensure that $\hat{H}_{\mathrm{meas}}$ is
defined independent of the choice of the background time $t$. Under the
hypotheses explained in the previous paragraph, the asymptotic future
values of $\hat{S}$ and $\hat{T}$ operators can be approximated as
\begin{align}
	\hat{S}^+
		&= \lim_{t\to\oo} \hat{S}(t) \\
		&\approx \hat{S}_{\mathrm{before}}
			+ \int\d{t}\, \delta[\hat{\tau}_l(t)-s]\dot{\hat{\tau}}_l(t) \hat{\tau}_l(t) \\
	T^+
		&= \lim_{t\to\oo} \hat{T}(t) \\
		&\approx \hat{T}_{\mathrm{before}}
			+ \int\d{t}\, \delta[\hat{\tau}_l(t)-s]\dot{\hat{\tau}}_l(t), \\
\notag
		&\qquad {}\times
					\int\d{t'}\d\mathbf{n}\,
						\delta^3[\hat{z}_{(t',\mathbf{n})}(t)-\hat{y}_l(t)]
							\hat{\tilde{\tau}}_{(t',\mathbf{n})}(t).
\end{align}
Provided $\hat{S}_{\mathrm{before}}$ and $\hat{T}_{\mathrm{before}}$ have zero
expectation value, the measurements of $\hat{S}^+$ and $\hat{T}^-$ the above
asymptotic limits provide unbiased estimates of the remaining terms,
which can be interpreted respectively as the reception and emission
times defined in the time delay experimental protocol.

We conclude this analysis by noting that, provided that any potential
uncertainties can be neglected or modeled and subtracted, it is
reasonable to assume that the spectral density of the quantum
gravitational vacuum with respect to $\hat{S}$ can be well-approximated by
\begin{equation}\label{S-spec-meas}
	\langle \delta(\hat{S}^+-\sigma) \rangle \approx \delta(\sigma-s) .
\end{equation}
It follows that it is then also reasonable to assume that the joint
spectral density of the quantum gravitational vacuum with respect to
$\hat{S}^+$ and $\hat{T}^+$ will be well-approximated by
\begin{equation}\label{ST-spec-meas}
	\rho(\sigma,t) =
	\langle \delta(\hat{S}^+-\sigma) \delta(\hat{T}^+-t) \rangle
		\approx \delta(\sigma-s) \rho_s(t) .
\end{equation}
The probability distribution $\rho_s(t)$ then has the interpretation of
the spectral density of the quantum gravitational vacuum with respect to
the quantum emission time operator observable $\hat{\tau}(s)$, where the
quantum emission time observable can then be identified as
$\hat{\tau}(s) = \hat{T}^+$. The quantum time delay observable is then
simply $\delta\hat{\tau}(s) = s-\hat{\tau}(s)$.

At this point it is worth considering a bit more precisely how the
probability distributions of Eqs.~\eqref{S-spec-meas}
and~\eqref{ST-spec-meas} relate to those that would be obtained by a
physically realized experiment following the same operational protocol.
Above, we have explicitly stated that these expressions are expected to
be applicable provided all sources of quantum (or classical)
fluctuations other than the quantum gravitational vacuum are neglected.
How can this neglect be reasonable? These neglected sources are
numerous, as for instance discussed in the later Sec.~\ref{fin-resol}.
Moreover, by now, the study of uncertainties in measurements of space
and time intervals induced by quantum fluctuations of the internal
states of a measurement apparatus is classical subject, going back to a
seminal paper of Salecker and Wigner~\cite{salecker-wigner}. These
effects can be quite large compared to the Planck-scale effects
(Sec.~\ref{dim-anl}) that we are concerned with here.

The main difference between the effects we neglect and the effect that
we actually study is that the former depend primarily on the internal
physics of the apparatus, while the latter depends crucially on the
dynamical quantum gravitational field. That is, the effect that we study
is genuinely due to quantum gravity, while those we neglect are not.  If
we are concerned with a question of principle, which is to account for
all possible sources contributing to the variance of the time delay or
time emission observables described earlier, it is not \emph{sufficient}
to include \emph{only} the internal apparatus sources or \emph{only} the
quantum gravitational effects (independent of their relative size), it
is in fact \emph{necessary} to consider both of them.  The internal
apparatus fluctuation sources have already been studied extensively in
the literature spawned by the original work~\cite{salecker-wigner}. On
the other hand, genuine quantum gravitational effects have received much
less attention (a review of the relevant literature was given in the
\hyperref[intro]{Introduction}) and are hence the main focus of the
current work. Since these contributions to the observational variance
are separate, they can be analyzed separately and, at the leading
perturbative order, contribute essentially additively.

When restored, the effects of quantum fluctuations of the internal
dynamics of the clocks and recording devices used in described
measurement models replace the sharp $\delta$-function in
Eq.~\eqref{S-spec-meas}, as well as in the analogous equation for
$\hat{T}^+$, by a broader probability distribution
$\rho_{\text{int}}(\sigma-s)$. The joint probability
distribution~\eqref{ST-spec-meas} would also be broadened broadened by
convolution with $\rho_{\text{int}}$ in both the $\sigma$ and $t$
arguments.

Finally, the calculations outlined in this paper do more than answer a
question of principle. As previously mentioned, they serve as a toy
model for resolving the challenges inherent in the problem of
observables in quantum gravity. So the lessons learned here, may be
applicable to a situation like early Universe cosmology, which is a more
likely source of physically measurable quantum gravitational
effects~\cite{woodard-qg}.

\section{Sketch of calculation in quantum linearized gravity}\label{quant-calc}
The point of the preceding section was to motivate that the output of an
explicit quantum calculation should be a probability distribution
$\rho_s(t)$, which should be interpreted as the spectral density of the time
register $\hat{T}$ with respect to the quantum gravitational vacuum (projected
onto the $s$-eigensubspace of the time register $\hat{S}$).
Phenomenologically, dropping the subscript $s$ since no other probability
distribution would be considered from now on, $\rho(t)$ should be interpreted as the statistical
distribution of measurement outcomes for an ensemble of repeated
measurements of the emission time $\hat{\tau}(s)$ (reproducing the geometry of
the experiment for each repetition). Again, motivated by the discussion
of the preceding section, we propose that, within the linearized gravity
approximation and keeping all available parameters tuned to minimize all
influences on the measurement of emission time other than the effects of
the gravitational field, the role of the observable $\hat{T}$ should be played
by the linearized expression~\eqref{emmtime-lin} with the classical graviton field
everywhere replaced by the quantized graviton field (with one caveat to
be discussed below) and the role of the quantum gravitational vacuum
should be played by the Poincar\'e-invariant Fock vacuum of the graviton
field. Within this proposal, the probability distribution $\rho(t)$ can be
computed explicitly. We leave the details of this calculation to be
presented elsewhere~\footnote{B.\ Bonga and I.\ Khavkine  (in
preparation).} and only discuss some general aspects of it that can
be deduced from dimensional analysis and the nature of perturbative
calculations.

\subsection{Gaussian spectral density}
In linearized gravity on Minkowski space, we interpret the Poincar\'e-invariant
Fock vacuum $|0\rangle$ as the quantum gravitational vacuum
and the operator $\hat{\tau}(s)$ as the quantum emission time. From the
preceding discussion, our goal is to evaluate the spectral
density $\rho(t)$ of $|0\rangle$ with respect to $\hat{\tau}(s)$. This
simplified problem has an explicit solution. A linear field theory is
essentially a collection of harmonic oscillators and, by construction,
the Fock vacuum is a Gaussian state (with zero mean) with respect to the
oscillator variables. In other words, the Fock vacuum is also Gaussian
with respect to any observable linear in the graviton
field $\hat{h}(x)$, such as $\hat{\tau}(s)$. Therefore, the sought
probability distribution $\rho(t)$ is Gaussian. It is fully determined
by its mean, which is just the classical Minkowski space expression~\eqref{emmtime-mink},
and its variance, which can be obtained from the expectation value
$\langle0|\hat{\tau}(s)^2|0\rangle$.

It is clear that the calculation of the probability distribution
$\rho(t)$ is reduced to evaluating a single vacuum expectation value
given above. Recall that the emission time is invariant with respect to
gauge transformations that fix the synchronization point $O$
and the lab tetrad frame at it. However, due to the Poincar\'e
invariance of the Fock vacuum, the expectation value
$\langle0|\hat{\tau}(s)^2|0\rangle$, which combines the observable and
the state, is actually invariant under arbitrary gauge transformations
and no longer depends on the special choice of synchronization point or
lab tetrad frame.

\subsection{Causal inequalities and perturbation theory}
In Sec.~\ref{caus-ineq}, we found that the time delay and emission time
observables obey some causal inequalities. The validity of these
inequalities relies mainly on the Lorentzian character of the metric
tensor and the geodesic character of inertial motion. Therefore,
classically, violations of these inequalities would be evidence of
superluminal signal propagation or violation of the equivalence
principle. Thus we can naturally take the following objective criterion
for the presence of causality violation in quantum theory: violation of
causal inequalities by the spectral density of the quantum gravitational
with respect to the time delay observable.

Quantum theory is famous for tunneling phenomena. For example, a quantum
state may be such that a measurement may find a particle (though likely
with only small probability) in a region that is classically forbidden to
it. Similarly, we would like to investigate whether the causal
inequalities are strictly obeyed in the quantum theory or are subject to
violations via ``quantum tunneling.'' A definite answer to this question
would go a long way toward informing the debate on whether \emph{any}
quantum theory of gravity \emph{necessarily entails} causality
violations~\cite{wheeler-foam, kent-caus}. 
While it would be very difficult to settle this
debate, in large part due to the breadth of the subject matter, as
stated. However, an explicit example of a quantum gravitational theory
without causality violation would force a weakening of the ``necessarily
entails'' clause. Equally, a fairly conservative (no extra matter, no
extra dimensions, no causality violation in the classical limit, though
taken only in a linear approximation) example
of a quantum gravitational model with causality violation would
strengthen the evidence for the ``any'' clause.

Unfortunately, as should become immediately obvious, the perturbative
calculations outlined in this section are not conclusive enough to
establish whether causality violation actually takes place or not. In
short, since the spectral density is expected to be Gaussian (as
discussed in Sec.~\ref{quant-model}),
\begin{equation}
	\rho(t) \sim \exp\left(-\frac{(t-\tau_\mathrm{cl})^2}{2(\Delta \tau)^2}\right),
\end{equation}
with some mean $\tau_\mathrm{cl}$ and variance $(\Delta \tau)^2$,
all real values of $t$ acquire a non-zero probability of being measured.
Thus, the causal bounds on $t$ are clearly violated, as illustrated in
Fig.~\ref{spec-dens-gaus-width}. However, the
responsibility for this violation can be ultimately traced back to the
perturbative approximation rather than to the quantum theory. Recall
that classically, as discussed in Sec.~\ref{caus-ineq}, the proofs of these
causal bounds crucially relied on the Lorentzian character of the
metric, as well as on the detailed behavior of geodesics in Lorentzian
spacetimes. Neither of these properties survives in perturbation theory.
One can find classical field configurations of $h_{ij}(x)$ for which the
linearized classical expression for $\tau(s)$ violates the causal
inequalities as well. For these field configurations $h_{ij}(x)$ would
have to be of the same order as the background metric $\eta$, which is
precisely the regime where perturbation theory is no longer applicable.

\begin{figure}
\includegraphics[width=\columnwidth]{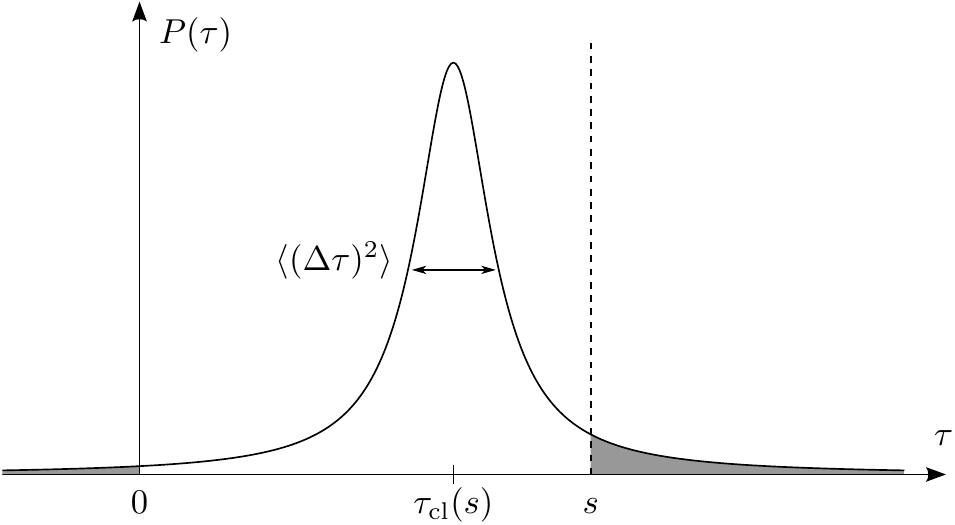}
\caption{%
Spectral density of the linearized quantum gravitational vacuum with
respect to the emission time $\tau(s)$. In a linear field theory, the
spectral density is expected to be Gaussian, with mean
$\tau_\mathrm{cl}(s)$ and variance $\langle(\Delta\tau)^2\rangle$. The
mean is the emission time in Minkowski space, Eq.~\eqref{emmtime-mink}.
This probability distribution clearly penetrates the shaded region,
which is forbidden by classical causal inequalities.
}
\label{spec-dens-gaus-width}
\end{figure}

In conclusion, the perturbatively calculated $\rho(t)$ may be presumed
to give accurate results around the interval $[\tau_\mathrm{cl}-\Delta
\tau, \tau_\mathrm{cl}+\Delta \tau]$, but not for larger or smaller values.
Unfortunately, the information needed to decide whether causal
inequalities are actually violated requires the knowledge of $\rho(t)$
precisely in the regions where the perturbative approximation is no
longer expected to be valid.

\subsection{Finite measurement resolution}\label{fin-resol}
The detailed calculation of the vacuum fluctuation of $\hat{\tau}(s)$
immediately presents a problem: it is infinite. This infinity can be
traced back to the singularity of the two-point function
\begin{equation}
	G(x-y) = \langle \hat{h}(x) \hat{h}(y)\rangle
		\sim \frac{1}{(x-y)^2}
\end{equation}
in the coincidence limit $x\to y$. This
infinity has a straightforward physical interpretation, which at the
same time suggests a meaningful regularization of the divergence.

Any realistic measurement of the quantum field $\hat{h}(x)$ is
carried out by a detector with finite spatial and temporal resolution.
Thus, no measurement is ever sensitive directly to the field
evaluated at a single spacetime point $x$, rather measurements are
typically sensitive to smeared fields~\cite{br-meas, bs-meas}
\begin{equation}\label{hsmeared}
	\tilde{h}(x) = \llangle \hat{h}(x-z) \rrangle = \int\d{z}\, \hat{h}(x-z) g(z) ,
\end{equation}
where $g(z)$ is a smooth test function peaked in the neighborhood of
$0$, and $\llangle{\cdots}\rrangle$ denotes the smearing with respect to
$g(z)$. The smearing function $g(z)$ may be interpreted as the detector
sensitivity profile, which clearly depends on how the measurement was
carried out. The vacuum fluctuation of the smeared field at $x$ is then
always finite 
\begin{align}
	\langle \tilde{h}(x)^2 \rangle
		&= \int\d{z_1}\d{z_2}\, G(x-z_1+z_2) g(z_1) g(z_2) \\
		&= \int\d{z}\, G(x-z) \tilde{g}(z) \\
		&\sim \llangle \frac{1}{(x-z)^2} \rrangle \\
		&\sim \frac{1}{\mu^2} ,
\end{align}
where $\tilde{g}(z) = (g*g)(z)$ is the convolution of $g(z)$ with
itself, by abuse of notation $\llangle{\cdots}\rrangle$ also denotes
smearing with respect to $\tilde{g}(z)$, and $\mu$ is the length scale
over which $g(z)$ has appreciable support, which is the spatiotemporal
resolution of the detector.  Physically, this estimate means that the
root-mean-square noise in a detector, due to quantum fluctuations, grows
as inversely proportional to its resolution (\cite{br-meas}
and Secs.~10.9.1--2 of~\cite{mw-opt}). Such fluctuations are vividly
illustrated in the context of quantum optics
in Fig.~2.1 of~\cite{leonhardt}.

Since we are working with an idealized model of physical measurement, it
is natural that the quantum fields entering into the expression for the
emission time $\hat{\tau}(s)$ should be smeared. Unfortunately, the
details of precisely how the smearing is to be done are quite
complicated. They in general depend on all the aspects of the
experiment: the resolutions of the proper time clocks, the coupling of
the lab and probe centers of mass to the gravitational field in geodesic
motion, sharpness of the signals transmitted by the probe, etc. For the
purposes of this discussion, we do not need such detailed information,
as for simplicity we would only be interested in the asymptotic limit of
perfect detector resolution $\mu\to 0$. This limit is obviously
divergent, so we can settle for the leading term in an expansion in
inverse powers of $\mu$.
Therefore, we simply assume that
all occurrences of the point field $\hat{h}_{ij}(x)$ are replaced by the
smeared field $\tilde{h}(x)$, Eq.~\eqref{hsmeared}. That is, the smearing
function $g(z)$ is the same everywhere, independent of $x$. The only
thing we assume about $g(z)$ is that it is regular enough to render the
vacuum fluctuation of $\hat{\tau}(s)$ finite and that it is peaked only
at the origin, with appreciable support over a region of size $\mu$,
so that we can estimate its moments as
\begin{equation}
	\llangle z^k \rrangle \sim \mu^k,
\end{equation}
where $z^k$ represents any homogeneous expression of order $k$ in the
components of $z$. It is worth noting that, as stated, this smearing
convention breaks background Lorentz invariance. This is clearly
unphysical. Nevertheless, we make this assumption in the current and
some future calculations for the purposes of working out their general
structure. A more physical smearing convention should be re-examined in
the future alongside with more realistic models of lab, probe, and
signal subsystems.

It is worth noting at this point that the works of Ford \emph{et
al.}~\cite{ford-lightcone, ford-top, ford-focus, ford-angle} took a
completely different approach to the regularization of divergences
arising from the singularities of the graviton two-point function. In
particular, they treated several scenarios that produced fluctuations
different from Minkowski space (finite temperature state, squeezed
vacuum, extra compactified dimensions), which were regularized by
subtracting the divergent Minkowski, Poincar\'e-invariant vacuum result.
Thus these previous calculations computed the deviation of the quantum
fluctuations from that of Minkowski space, but did not directly address
Minkowski space results themselves, unlike we do in this work.

\subsection{Dimensional analysis}\label{dim-anl}
Looking at the structure of the explicit expression for the linearized
correction $r\tau_\mathrm{cl}(s)$ to the emission $\tau(s)$,
Eq.~\eqref{emmtime-lin}, it is fairly obvious that a detailed
calculation of the variance $\langle \tilde{r}^2 \rangle$ of the smeared
correction $\tilde{r}$, where each occurrence of the classical field
$h(x)$ is simply replaced by the smeared quantum field $\tilde{h}(x)$,
will be quite involved. The expression for $r$, whose structure is
illustrated at the end of Sec.~\eqref{triang-lin}, contains on the order of $10$
terms. Therefore, the number of terms in $\tilde{r}^2$ will be of order
$100$. Each of these terms consists of two (possibly iterated) integrals
over spacetime segments over (possibly iterated) derivatives of the
smeared $\llangle G(z) \rrangle$ graviton two-point function. The total
number of nested integrations for each term is five (5), which includes two (2)
from the spacetime segment and three (3) from smearing. Using symmetry, one or
two integrations may be made trivial.  However, it is unavoidable that
each of the order $100$ is a high-dimensional integral. Moreover, the
integrands are distributions, rather than continuous functions, whose
singularities are ultimately traceable to the light-cone and coincidence
singularities of the graviton two-point functions. The high
dimensionality of the integrals and the distributional character of the
integrands makes it very difficult to treat them numerically. On the
other hand, the integrands of these order $100$ terms may have many
different algebraic structures, preventing the evaluation of a single
master analytical expression that could be uniformly applied to all of
them. Splitting each term into simpler pieces and considering all
possible cases of algebraic structures easily leads to thousands of
individual integrals to be evaluated analytically.  There is little
choice but to resort to hybrid numerical-analytical calculations
automated using computer algebra software. These detailed calculations
are in progress and their results will be reported elsewhere~[32]. In the
rest of this section we concentrate on some intermediate, qualitative
results that may be obtained by straightforward dimensional analysis.

Taking dimensionful constants into account, and keeping in mind that the
field $h(x)$ is itself dimensionless, the unsmeared graviton two-point
functions has the form
\begin{equation}
	\langle \hat{h}(x) \hat{h}(y) \rangle = G(x-y) \sim \frac{\ell_p^2}{(x-y)^2},
\end{equation}
where the denominator of the last expression is the spacetime interval
squared, $(x-y)^2 = \eta_{ij} (x-y)^i(x-y)^j$, and the numerator is the
Planck length squared, $\ell_p^2 = G\hbar/c^3$. What is important
here is that $G(z)/\ell_p^2\sim 1/z^2$ is a homogeneous function of $z$ of degree $-2$ and hence
of length dimension $[G(z)/\ell_p^2]=-2$. The scales $\mu$ and $\ell_p$
and the components of $z$ itself all have length dimension
$[z]=[\mu]=[\ell_p] = 1$. On the other hand, a derivative with respect
to $z$ has length dimension $[\nabla]=-1$. Generically it has the effect
$\nabla z^n \sim z^{n-1}$. Using the convention from the
\hyperref[pertsol]{Appendix}, the spacetime segment integrals are all
affinely parametrized from $0$ to $1$ and hence are dimensionless,
$[\int_X] = 0$. On the other hand, integration over a spacetime segment
has the generic effect $\int_X z^{n} \sim z^{n+1}/s$, where $s$ is the
length scale of the segment $X$, $[s]=1$, and $z$ on the right hand side
corresponds to the coordinates of the segment's end points.

Without smearing, the expectation value $\langle r^2 \rangle$ is
infinite. Smearing introduces a regulating length scale $\mu$, the
detector resolution. Therefore, the smeared expectation value $\langle
\tilde{r}^2 \rangle$ should diverge as $\mu\to 0$. The details of the
approach of $\mu$ to $0$ in general depend on the details of the
smearing functions. Fortunately, a kind of universality among all
well-behaved localized smearing functions can be obtained by
concentrating on the leading terms in an expansion of the result in
inverse powers of $\mu$.

From the structure of the explicit expression for $r$, keeping in mind
that derivatives worsen singularities while integrals improve them, the
most singular contribution should come from the terms with the greatest
number of derivatives and the least number of integrals. Namely, $r \sim
r_X \int_X \nabla h$, where $r_X$ is some tensorial coefficient
dependent on the geometry of the segment $X$. Note that, since both $r$
and $h$ are dimensionless, the tensorial coefficient $r_X$ must have
length dimension $[r_X] = 1$ and be of order $s$ in magnitude, due to
the standard affine parametrization of the integral over $X$. In fact,
it should be of size $s$, which is the length scale of the spacetime
segment $X$. The leading-order contribution to the smeared variance of
$r$ can then be estimated as follows:
\begin{align}
	\langle \tilde{r}^2 \rangle
	&\sim \left\langle \left( \sum_X r_X \int_X \nabla \tilde{h} \right)^2
		\right\rangle \\
	&\sim s^2 \llangle \int_X \int_Y \nabla^2 G(z) \rrangle \\
	&\sim s^2 \llangle \int_X \int_Y \nabla^2 \frac{\ell_p^2}{z^2} \rrangle \\
	&\sim s^2 \llangle \int_X \int_Y \frac{\ell_p^2}{z^4} \rrangle \\
	&\sim s^2 \llangle \frac{\ell_p^2}{s^2 z^2} \rrangle
		\sim \llangle \frac{\ell_p^2}{z^2} \rrangle 
		\sim \frac{\ell_p^2}{\mu^2}
\end{align}
Detailed calculations show that many terms do have this scaling
behavior, but also that terms of the form $(\ell_p^2/\mu^2)\log(\mu/s)$
and $(\ell_p^2/\mu^2)(s/\mu)$ show up at intermediate stages as well.
While the appearance of logarithmic scaling is not unusual in quantum
calculations, the last term is somewhat surprising and, if uncanceled
in the final result, may cast serious doubt on the validity of the
linearized approximation in the regimes of very large $s/\mu$ ratios.
This ratio corresponds to that of the spatial and temporal extent of the
experiment to the resolution of the detectors involved.

From Eq.~\eqref{emmtime-lin}, the perturbative correction to the
emission time $\tau(s)$ and the time
delay $\delta\tau(s)$ scale like $s r$, and so the quantum variances of
$\hat{\tau}(s)$ and $\delta\hat{\tau}(s)$ should scale like
$s^2\langle\tilde{r}^2\rangle$, since $\langle\tilde{r}\rangle=0$. From this
and the possible leading-order contributions to the smeared variance of
$r$ we can deduce the root-mean-square size of
fluctuations expected in observations of the time delay due to the
fluctuations of the quantum gravitational vacuum shown in Table~\ref{rms-est}.
Let us contrast two possible experimental contexts. In the
\emph{laboratory} context, the spatiotemporal extent of the experiment
(with time-length conversion via the speed of light) is expected to be
$s\sim 1~\mathrm{m} \sim 10^{-9}~\mathrm{s}$, while in the
\emph{cosmological} one $s\sim 1~\mathrm{Mpc}$. Recall that a megaparsec
is $1~\mathrm{Mpc} \sim 10^{22}~\mathrm{m} \sim 10^{14}~\mathrm{s}$. For
the detector resolution scale, we select $\mu\sim 1~\mathrm{nm} \sim
10^{-18}~\mathrm{s}$. This is of the order of the wavelength of X-rays,
which are consistently available in both contexts. The Planck scale
as usual is $\ell_p \sim 10^{-35}~\mathrm{m} \sim 10^{-44}~\mathrm{s}$.
\begin{table}[h]
\caption{%
	Estimates for the root-mean-square size of quantum fluctuations in observations of
	the time delay for different possible leading-order behaviors in
	$\mu\to 0$.
}
\label{rms-est}
\begin{ruledtabular}
\begin{tabular}{cccccc}
Context & $s$ & $\mu$ & $\frac{s\ell_p}{\mu}$
	& $\frac{s\ell_p}{\mu}\log(\frac{s}{\mu})^{1/2}$
	& $\frac{s^{3/2}\ell_p}{\mu^{3/2}}$\\
\hline
laboratory
	& $1$ m & $1$ nm & $10^{-35}$ s & $10^{-34}$ s & $10^{-30}$ s \\
cosmological
	& $1$ Mpc & $1$ nm & $10^{-12}$ s & $10^{-11}$ s & $10^3$ s
\end{tabular}
\end{ruledtabular}
\end{table}

All of the above estimates, except one, are well below the sensitivity
or noise thresholds of the current state of the art of experimental and
observational technology. So it is not surprising that kind of effect
has yet to be observed.  Clearly, if the largest of the above estimates
were correct, we would have observed this effect long ago due to the
very large fluctuations in the arrival times of high frequency photons
from distant galaxies. Of course, since that result is only preliminary
and comes from the least understood part of intermediate calculations,
it has to be taken with a grain of salt. But it does highlight the fact
that the linearized approximation employed in the calculations described
above may not be valid on large timescales. This is not an unusual
feature of perturbation theory. For example, it was noticed long ago in
celestial mechanics that there exist perturbative terms that scale with
positive powers of time, so-called \emph{secular terms}, in otherwise
non-perturbatively stable systems~\cite{giacaglia}. This remark also offers some
hope that if, in fact, the perturbation expansion in our calculations
breaks down on large time scales that this problem could be repaired
using the methods already developed for dealing with secular
perturbative terms in celestial mechanics or other fields.


\section{Discussion}\label{discussion}
We have operationally defined a particular physical observable, the time
delay $\delta\tau(s)$ [as well as the related emission time $\tau(s)$],
and have provided both exact, implicit and approximate, explicit
mathematical models for it. The time delay satisfies two important
inequalities (stemming from the maximality of light speed and from local
geodesic extremality) directly related to the causal structure of
classical Lorentzian spacetimes. Thus, it is sensitive to the causal
structure of classical dynamical gravity. Moreover, we have sketched how
the same operational definition can be used to define a quantum time
delay observable and how to compute its variance due to quantum
fluctuations of the quantum gravitational vacuum, in linearized gravity,
given the usual Fock quantization of the graviton field.

This work opens up many potential lines of investigation. Foremost among
them, is the completion of the detailed calculation of the variance of
the time delay due to the quantum fluctuations of the quantum
gravitational vacuum. That work is in progress and will be reported on
elsewhere~[32].

An important issue that needs to be explored is the detailed
construction of a quantum model of the measurement apparatus sketched in
Sec.~\ref{quant-model}.  This model should take into account the quantum dynamics
of the center of mass motions of the probe and laboratory, a more
detailed representation of the time stamped signal transmitted by the
probe, and of the weak measurements of the relevant clock and signal
systems.  Some existing literature may be helpful in refining these
models~\cite{glcm-spread, gppt-montevideo, pw, bk-meas}.

The triangular geometry of the time delay experiment is one of the
simplest possible. However, there is no conceptual obstacle to
generalizing the same methodology to more complex geometries, including
piecewise geodesic motion with more components and even accelerated motion.
It is also natural to capture other effects of the fluctuating
gravitational field on the signal, such as angular blurring and other
image effects at the reception of the signal by the lab. These effects
were previously considered in~\cite{ford-angle}, though with caveats similar
to those given in the Introduction while discussing~\cite{ford-lightcone}.

It is clear that a whole class of physical observables of manageable
mathematical complexity and with clear physical interpretation can be
constructed using the same methodology. This class can be aptly named
\emph{astrometric observables} or \emph{quantum astrometric
observables}, when referring to them in the quantum context.

Yet another important generalization is to background geometries other
than Minkowski space. Cosmological and black hole backgrounds are of
particular importance. For instance, a similar calculation could model
the fluctuation in the arrival time of photons from distant galaxies due
to the intrinsic quantum fluctuation in the cosmological quantum state
of the graviton field. Such fluctuations would contribute to the spread
of the arrival times of photons from distant $\gamma$-ray bursts~\cite{grb}.
Undoubtedly, the final observational data compounds many effects,
including the likely more dominant astrophysical ones and those due to
in transit scattering. However, a thorough understanding of quantum
fluctuations in astrometric observables in linearized gravity (or
related approximations) is necessary before the observational data could
be used to infer the existence of exotic effects like violation of local
Lorentz invariance, spacetime discreteness or granularity, modified
dispersion relations, etc.~\cite{grb, ac-grb, hs-pheno}, since the model of quantum gravity
considered in the present calculation exhibits none of these features.
Also, the behavior of light signals and inertial or accelerated probes
in the vicinity of a black hole can be used to give an operational
meaning to the location of its horizon. The fluctuations of some quantum
astrometric observables could then be used to unambiguously study the
inferred quantum fluctuations of the black hole horizon.

A limitation of the proposed method of calculating the quantum vacuum
fluctuation of the time delay (or any other astrometric observable) in
quantum linearized gravity is the inability of perturbation theory to
address questions involving strong fields, like the question of whether
the quantum theory respects or violates the causal inequalities
discussed in Sec.~\ref{caus-ineq}. Unfortunately, in the physically
relevant case of four-dimensional spacetime, the only effective
calculational tool we have is perturbative quantum field theory.
Perhaps an improved perturbation theory in the spirit of the Magnus
expansion~\cite{bcor-magnus} can be used to keep the signature of the
metric tensor Lorentzian while still using perturbative methods, so that
the causal inequalities are not immediately violated already at the
classical level.  On the other hand, the time delay and astrometric
observables in general can be defined equally well in any spacetime
dimension. This opens up the possibility of adapting the quantum
calculation to the two- and three-dimensional versions of general relativity,
which can be solved exactly. A family of classical observables of
3-dimensional gravity that could be said to fall into the astrometric
category have been identified and expressed in variables that are
appropriate for treatment in the quantum theory in~\cite{meus}. The
quantum calculations have yet to be carried out.

Finally, since astrometric observables are defined in a way independent
of the underlying model of quantum gravity, their behavior could in
principle be studied in any of the popular (or even not so popular)
proposed theories of quantum gravity. It is often the case that it is
difficult to compare calculations between these different theories, due
to the very different underlying mathematical frameworks. It would be
very interesting to see if quantum astrometric observables can serve as
a benchmark suite to compare the predictions of each of these theories
on equal footing. 

\begin{acknowledgments}
The author would like to thank Renate Loll, Albert Roura, Sabine
Hossenfelder and Paul Reska for their support and helpful discussions.
The author also acknowledges support from the Natural Science and
Engineering Research Council (NSERC) of Canada and from the Netherlands
Organisation for Scientific Research (NWO) (Project No.\ 680.47.413).
\end{acknowledgments}

\appendix*

\section{Perturbative solution of geodesic and parallel transport equations}
\label{pertsol}
Let $e^a_i$ be a tetrad field, as described in Sec.~\ref{tetform}. Let
$\gamma(t)$ be a parametrized spacetime curve and $v_\alpha^a(t)$,
$\alpha=0,1,2,3$, an orthonormal tetrad along it. Its components
$v^i_\alpha(t)$ in the basis of the spacetime tetrad are given by
$v^a_\alpha(t) = v^i_\alpha(t) e^a_i(\gamma(t))$. The pair
$(\gamma,v^a_\alpha)$ is a geodesic with a parallel-transported
orthonormal frame on it if it satisfies the following conditions
\begin{align}
	\dot{\gamma}(t)^a &= v^a_0(t), \\
	\dot{\gamma}(t)^a \nabla_a v^c_\alpha(t) &= 0.
\end{align}
When the spacetime dual tetrad field is expressed in terms of a reference
inertial coordinate dual tetrad $\hat{x}^i_a$ (Eq.~\eqref{tetref}) as $e^i_a =
T^i_j \hat{x}^j_a$, the
geodesic and parallel transport equations are expressed in tetrad
components as follows
\begin{align}
\label{geoeq}
	\dot{\gamma}^i &= v^a_0 \hat{x}^i_a = v^j_0 \bar{T}^i_j,\\
\label{pteq}
	\dot{v}^k_\alpha &= -v^i_0 \omega\indices{_i^k_j} v^j_\alpha,
\end{align}
where $\eta_{kk'}\omega\indices{_i^{k'}_j} = \omega_{ikj} =
\omega_{i[kj]}$ are the Ricci rotation coefficients (Sec~3.4b
of~\cite{wald}). The Ricci
rotation coefficients can be computed in terms of the transformation
matrix $T^i_j$. Below, $\del_a=\hat{x}^i_a \del_i$ denotes the
coordinate derivative, $\Gamma^c_{ab}$ the usual Christoffel tensor,
encoding the difference between $\nabla_a$ and $\del_a$, and
$\Gamma_{cab} = g_{cc'} \Gamma^c_{ab}$.
\begin{align}
	\omega_{ikj}
	&= e_{kc} e_i^a \nabla_a e_j^c \\
	&= e_{kc} e_i^a \del_a e_j^c + e_{kc} e_i^a \Gamma^c_{ab} e_j^b \\
	&= e_{kc} e_i^a \del_a e_j^c + e_i^a e_j^b e_k^c \Gamma_{cab}.
\end{align}
Each term on the right hand side is evaluated separately below and
expressed in terms of a single quantity $\alpha_{ikj}$.
\begin{align}
	\alpha_{ikj}
	&= \bar{T}^{i'}_i (\del_{i'} T^{l}_{j'}\eta_{lk}) \bar{T}^{j'}_{j} \\
	e_{kc} e_i^a \del_a e_j^c
	&= \bar{T}^{i'}_i (\del_a \bar{T}^{j'}_j) T^l_{l'} \eta_{lk}
		\hat{x}^{l'}_c \hat{x}_{i'}^a \hat{x}_{j'}^c \\
	&= \bar{T}^{i'}_i (\del_{i'} \bar{T}^{l'}_j) T^l_{l'} \eta_{lk} \\
	&= -\bar{T}^{i'}_i \bar{T}^{j'}_{j}(\del_{i'} T^{l''}_{j'})
		\bar{T}^{l'}_{l''} T^l_{l'} \eta_{lk} \\
	&= -\bar{T}^{i'}_i (\del_{i'} T^{l}_{j'}\eta_{lk}) \bar{T}^{j'}_{j}
		= -\alpha_{ikj} , \\
	e_i^a e_j^b e_k^c \Gamma_{cab}
	&= \frac{1}{2} e_i^a e_j^b e_k^c (\del_a g_{bc}+\del_b g_{ac}-\del_c g_{ab})\\
	&= \frac{1}{2} \bar{T}^{i'}_i \bar{T}^{j'}_j \bar{T}^{k'}_k \\
\notag & \qquad {} \times
		\left[ \del_{i'}g_{j'k'} + \del_{j'}g_{i'k'} - \del_{k'}g_{i'j'} \right] \\
	&= \frac{1}{2} \bar{T}^{i'}_i \bar{T}^{j'}_j \bar{T}^{k'}_k
		\left[
			\del_{i'}(T^{J}_{j'} \eta_{JK} T^{K}_{k'}) \right. \\
\notag & \qquad \left. {} +
			\del_{j'}(T^{I}_{i'} \eta_{IK} T^{K}_{k'}) -
			\del_{k'}(T^{I}_{i'} \eta_{IJ} T^{J}_{j'})
		\right]
		\\
	&= \frac{1}{2}\left[
			  \alpha_{ikj} + \alpha_{ijk}
			+ \alpha_{jik} \right. \\
\notag & \qquad \left. {}
			+ \alpha_{jki}
			- \alpha_{kij} - \alpha_{kji}
		\right] , \\
	\omega_{ikj}
	&= -\alpha_{ikj} + \frac{1}{2}\left[
			  \alpha_{ikj} + \alpha_{ijk}
			+ \alpha_{jik} \right. \\
\notag & \qquad \left. {}
			+ \alpha_{jki}
			- \alpha_{kij} - \alpha_{kji}
		\right] \\
	&= \frac{1}{2}\left[
			- \alpha_{ikj} + \alpha_{ijk}
			+ \alpha_{jik} \right. \\
\notag & \qquad \left. {}
			+ \alpha_{jki}
			- \alpha_{kij} - \alpha_{kji}
		\right] \\
	&= -3\alpha_{[ikj]} + 2\alpha_{[j|i|k]} \\
	&= -\alpha_{i[kj]} + \alpha_{j(ik)} - \alpha_{k(ij)}.
\end{align}
The alternative expressions for $\omega_{ikj}$ in terms of
$\alpha_{ikj}$ are provided for convenience. When $T=\exp(h)$, and $h$
is considered to be small, the linear-order expression for $\alpha$ in
terms of $h$ is
\begin{equation}
	\alpha_{ikj} = \del_i h^l_j \eta_{lk} + \O(h^2)
		= \del_i h_{kj} + \O(h^2) .
\end{equation}

\begin{figure}
\includegraphics{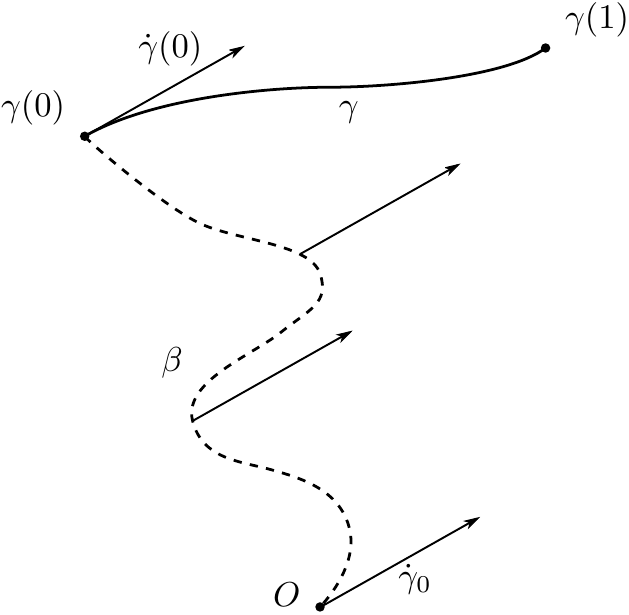}
\caption{%
A geodesic $\gamma$ is defined by its initial point $\gamma(0)$ and
initial tangent vector $\dot{\gamma}(0)$. The initial point itself is
specified as the final point $\gamma(0)=\beta(1)$ of another curve
$\beta$ which starts at the origin. The initial tangent vector can then
be specified by its inverse image $\dot{\gamma}_0\in T_OM$ under
parallel transport along $\beta$.
}
\label{geodesic-geom}
\end{figure}

The geodesic~\eqref{geoeq} and parallel transport~\eqref{pteq} equations
can be jointly transformed into a system of integral equations
\begin{align}
	\hspace{-.3em}
	\gamma(t)^i
		&= \gamma(0)^i + \int_{0}^t\d{t'}\, \bar{T}(\gamma(t'))^i_j v_0^j(t'), \\
	\hspace{-.3em}
	v^k_\alpha(t)
	&= T\exp\left[
			-\int_0^t\d{t'}v_0(t')^i\omega(\gamma(t'))\indices{_i^k_j}
		\right] v_\alpha^j(0), \\
\label{expp-def}
	&= \exp(p_\gamma(t))^k_j v_\alpha^j(0),
\end{align}
where $T\exp({\cdots})$ denotes the time-ordered exponential and the
\emph{parallel propagator} $\exp(p_\gamma(t))^k_j$ is defined implicitly
by the last equation. For brevity, we also use the notation $p_\gamma =
p_\gamma(1)$. In this form, the solution can be directly expanded to any
desired order in $\O(h)$. The solutions are parametrized by the initial
data $\gamma(0)^i$ and $v_\alpha^i(0)$, with $\dot{\gamma}^i(0) =
v_0^i(0)$.

The initial data are specified as described in Sec.\ref{tetform}.
Namely, given a curve $\beta$ starting at the origin, $\beta(0)=O$, we
have $\beta^i(1) = \gamma^i(0)$ and $v_\alpha^i(0) = \exp(p_\beta)^i_j
L^j_k v_{O,\alpha}^k$, for some vectors $v_{O,\alpha}^k \hat{e}^a_k
= L^j_k v_{O,\alpha}^k e^a_j$ in the tangent space at $O$.

Suppose that at zeroth order we are given $\beta(t) = \beta_0(t) +
\beta_1(t) + \O(h^2)$, $\gamma(t) = \gamma_0(t) + \O(h)$, and
$\exp(p)^k_j = \delta^k_j + \O(h)$. To linear order, the parallel
propagator is expanded as
\begin{align}
	\exp(p_\gamma(t))^k_j &= \delta^k_j + (H_{\gamma_0}(t))^k_j + \O(h^2), \\
	(H_{\gamma_0}(t))_{kj} &= \eta_{ki} (H_{\gamma_0}(t))^i_j \\
\label{H-apdx}
		&= -\int_0^t\d{t'}\,
			\dot{\gamma}_0^i(t') \omega(\gamma_0(t'))_{ikj}, \\
	H_{\gamma_0} &= H_{\gamma_0}(1).
\end{align}
Recall that $L=\exp(h_O)$. The tetrad components of the parallel-transported
vector $v_\alpha^a$ at $\gamma(1)$ are then
\begin{align}
	v_\alpha^i(1)
	&= \exp(p_\gamma)^i_j \exp(p_\beta)^j_k \exp(h_O)^k_l v_{O,\alpha}^l \\
\label{prlt-sol}
	&= v_{O,\alpha}^i + (h_O)^i_j v_{O,\alpha}^j
		+ (H_{\beta_0}+H_{\gamma_0})^i_j v_{O,\alpha}^j \\
\notag
	&\qquad {} + \O(h^2) .
\end{align}
For the coordinates of the geodesic curve, the linear-order solution is
[recall that $\bar{T}=\exp(-h)$ and $L=\exp(h_O)$]
\begin{widetext}
\begin{align}
	\gamma^i(1)
	&= \beta^i(1) + \int_0^1\d{t'}\,
		\exp(-h(\gamma(t')))^i_j \exp(p_\gamma(t'))^j_k
		\exp(p_\beta)^k_l \exp(h_O)^l_m v_{O,0}^m \\
	&= \beta_0^i(1) + \beta_1^i(1) + \int_0^1\d{t'}\,\left[
		- h^i_j(\gamma_0(t'))
		+ (H_{\gamma_0}(t'))^i_j
		+ (H_{\beta_0})^i_j
		+ (h_O)^i_j \right] v_{O,0}^j + \O(h^2) \\
	&= \beta_0^i(1) + \beta_1^i(1)
		- \int_0^1\d{t'} h^i_j(\gamma_0(t')) v_{O,0}^j
		- \int_0^1\d{t'} \int_0^{t'}\d{t''}
				\dot{\gamma}_0^k(t'')\omega(\gamma_0(t''))\indices{_k^i_j} v_{O,0}^j \\
\notag
	&\qquad {}
		+ (H_{\beta_0})^i_j v_{O,0}^j
		+ (h_O)^i_j v_{O,0}^j
		+ \O(h^2), \\
\label{geod-sol}
	&= \beta_0^i(1)+\beta_1^i(1)+(J_{\gamma_0,\beta_0})^i+\O(h^2),\\
	(J_{\gamma_0,\beta_0})_i &= \eta_{ij} (J_{\gamma_0,\beta_0})^j \\
\label{J-apdx}
	&=- \int_0^1\d{t'} h_{ij}(\gamma_0(t')) v_{O,0}^j
		- \int_0^1\d{t'} \int_0^{t'}\d{t''}
				\dot{\gamma}_0^k(t'')\omega(\gamma_0(t''))_{kij} v_{O,0}^j
		+ (H_{\beta_0})_{ij} v_{O,0}^j
		+ (h_O)_{ij} v_{O,0}^j .
\end{align}

For the purposes of this paper, all the zeroth-order curves $\beta_0(t)$
or $\gamma_0(t)$ are piecewise straight line segments in Minkowski
space. Given a straight line segment $X(t)$, we denote its point of origin
and end point by $x_1=X(0)$ and $x_2=X(1)$ respectively. The standard
affine parametrization is $X(t) = (1-t)x_1 + tx_2$. The segment's
standard tangent vector is denoted $x=x_2-x_1$. At zeroth order, the
$\beta$-parallel-transported image of $x^i$ at $O$ is just $v_{O,0}^i=x^i$.
It is convenient to use the following notation and identities for
integration over line segments:
\begin{align}
\label{apdx-not1}
	\int_X^{(n)} \d{t} f
	&= \int_0^1\d{t_0}\int_0^{t_0}\d{t_1}\cdots\int_0^{t_{n-1}}\d{t_n}\,
		f(X(t_n)) \\
	&= \int_0^1\d{t}\,\frac{(1-t)^n}{n!} f(X(t)) , \\
\label{apdx-not2}
	\int_X \d{t}\, f &= \int^{(0)}_X \d{t}\, f , \\
\label{apdx-not3}
	[f]_{x_1}^{x_2}
	&= \int_X \d{t}\, x^i \del_i f = f(x_2) - f(x_1) , \\
\label{apdx-not4}
	\int^{(n)}_X\d{t}\, x^i \del_i f
	&= -\frac{1}{n!} f(x_1) + \int^{(n-1)}_X\d{t}\, f .
\end{align}
For definiteness, suppose that $\gamma_0$ is a single segment $X$, whose
point of origin is the end point of $\beta_0$, a piecewise linear path
$Y$, whose segments are indexed by $N$ and denoted $Y_N$.  Then we can
write the expressions for $H_X$, $H_Y$ and $J_{X,Y}$ more concretely as
\begin{align}
	A^i B^j (H_X)_{ij}
	&= A^i B^j \int_X\d{t}\,x^k[\del_k h_{[ij]}-\del_j h_{(ki)}+\del_i h_{(kj)}]\\
\label{H-details}
	&= A^i B^j [h_{[ij]}]_{x_1}^{x_2}
		+ 2A^{[i}B^{j]} x^k \int_X\d{t}\, \del_i h_{(kj)}, \\
	(H_Y)_{ij} &= \sum_N (H_{Y_N})_{ij} ,
\end{align}
where the expression for $H_{Y_N}$ is the same, only with $X$ replaced by
$Y_N$, and
\begin{align}
	A^i (J_{X,Y})_i
	&= -A^i\int_X\d{t}\, h_{ij} x^j
		+A^i\int^{(1)}_X\d{t}\,
			x^k[\del_k h_{[ij]}-\del_j h_{(ki)}+\del_i h_{(kj)})]x^j
		+A^i(H_Y)_{ij} x^j + A^i (h_O)_{ij} x^j \\
\label{J-details}
	&= A^i x^j (h_O)_{ij} - A^i x^j h_{[ij]}(x_1)
		+ A^i x^j [h_{[ij]}]_{O}^{x_1} \\
\notag
	&\qquad {}
		- A^i x^j \int_X\d{t}\, h_{(ij)}
		+ 2 A^{[i}x^{j]} x^k \int^{(1)}_X \del_i h_{(kj)} 
		+ \sum_N 2 A^{[i}x^{j]}y_N^k \int_{Y_N}\d{t}\,\del_i h_{(kj)} \\
	&= A^i x^j [(h_O)_{ij} - h_{[ij]}(O)]
		- A^i x^j \int_X\d{t}\, h_{(ij)}
		+ 2 A^{[i}x^{j]} x^k \int^{(1)}_X\d{t}\,\del_i h_{(kj)} 
		+ \sum_N 2 A^{[i}x^{j]}y_N^k \int_{Y_N}\d{t}\,\del_i h_{(kj)} .
\end{align}

Given the above explicit formulas for $H_X$ and $J_{X,Y}$. We can consider
how they transform under gauge transformations. In the tetrad formalism,
gauge transformations are generated by arbitrary local Lorentz
transformations (tetrad rotations that do not change the metric) and
arbitrary spacetime diffeomorphisms. In the linearized tetrad formalism,
the most general gauge transformation takes the form $h_{ij} \mapsto
h_{ij} + \del_j C_i + D_{ij}$ and $(h_O)_{ij}\mapsto -\del_{[j}C_{i]} +
D_{ij}$, where $D_{ji}=-D_{ij}$. Under such a variation we have the
following identities:
\begin{align}
	\delta h_{ij} &= \del_j C_i + D_{ij}, \\
	\delta (h_O)_{ij} &= 2\del_{[j}C_{i]} + D_{ij}, \\
	A^i B^j (\delta H_X)_{ij}
	&= A^i B^j [\del_{[j}C_{i]} + D_{ij}]_{x_1}^{x_2}
		+ A^{[i}B^{j]} \int_X\d{t}\,[x^k\del_k \del_k C_j + x^k\del_i \del_j C_k]\\
	&= A^i B^j [\del_{[j}C_{i]}+D_{ij}]_{x_1}^{x_2}
		+A^{[i}B^{j]}[\del_i C_j]_{x_1}^{x_2} \\
\label{H-trans}
	&= A^i B^j [D_{ij}]_{x_1}^{x_2}, \\
	A^i (\delta J_{X,Y})_i
	&= A^i x^j[2\del_{[j}C_{i]}(O)-\del_{[j}C_{i]}(O)-B_{ij}(O)]
		- \frac{1}{2}A^i x^j \int_X\d{t}\,[\del_j C_i + \del_i C_j] \\
\notag
	&\qquad {}
		+ A^{[i}x^{j]}x^k\int^{(1)}_X\d{t}\,[\del_i\del_k C_j + \del_i\del_j C_k]
		+ \sum_N A^{[i}x^{j]}y^k_N\int_{Y_N}\d{t}\,
			[\del_i\del_k C_j + \del_i\del_j C_k] \\
	&= A^i x^j[\del_{[j}C_{i]}(O)-B_{ij}(O)]
		-\frac{1}{2}A^i [C_i]_{x_1}^{x_2}
		+A^{[i}x^{j]}[\del_i C_j]_{O}^{x_1}
		-\frac{1}{2}A^i x^j\int_X\d{t}\,\del_i C_j \\
\notag
	&\qquad {}
		- A^{[i}x^{j]}\del_i C_j(x_1)
		+ A^{[i}x^{j]}\int_X\d{t}\, \del_i C_j \\
	&= A^i x^j[\del_{[j}C_{i]}(O)-B_{ij}(O)]
		- A^{[i}x^{j]}\del_i C_j(O)
		- A^i [C_i]_{x_1}^{x^2}  \\
\label{J-trans}
	&= A^i x^j [2\del_{[j}C_{i]}(O)-B_{ij}(O)] - A^i[C_i]_{x_1}^{x_2} .
\end{align}
\end{widetext}
Suppose that we have a sequence of segments $(X)$ that starts at
$O$ and forms a closed loop, $\sum_X x^i = 0$. We can naturally form
pairs $(X,Y)$ where the segments $Y_N$ of $Y$ consist of elements of
$(X)$ that precede $X$. Then we can define
\begin{align}
\label{H-grp}
	(H_{(X)})_{ij} &= \sum_X (H_X)_{ij} , \\
\label{J-grp}
	(J_{(X)})_i &= \sum_X (J_{X,Y})_i .
\end{align}
From the above transformation properties, it is clear that $H_{(X)}$ and
$J_{(X)}$ are invariant under all
gauge transformations, $A^i B^j(\delta H_{(X)})_{ij} = 0$ and $A^i
(\delta J_{(X)})_i = 0$. The terms that depend only on the
values of $C$ and $D$ at $O$ cancel because their sum is proportional to
$\sum_X x^i$ and the remaining terms cancel because they form a cyclic
telescoping sum. These invariant quantities are actually used in
Sec.~\ref{triang-lin} in the expression for the correction to the time delay in
approximately Minkowski spacetime.

\bibliographystyle{apsrev4-1}
\bibliography{paper-delay}

\end{document}